\documentclass[graybox]{svmult}

\usepackage{mathptmx}       
\usepackage{helvet}         
\usepackage{courier}        
\usepackage[bottom]{footmisc}
\usepackage{graphicx}
\usepackage{amsmath,amssymb}
\usepackage{bm}

\newcommand{\cc}[1]{\overline{#1}}
\newcommand{\e}[0]{\E}
\newcommand{\bbset}[1]{\mathbb{#1}}
\newcommand{\N}[0]{\bbset{N}}
\newcommand{\Z}[0]{\bbset{Z}}
\newcommand{\R}[0]{\bbset{R}}
\newcommand{\C}[0]{\bbset{C}}
\newcommand{\Bs}[0]{\mathcal{B}}
\newcommand{\PWs}[0]{\mathcal{PW}}
\newcommand{\Fto}[0]{\mathcal{F}}
\newcommand{\Ft}[1]{\hat{#1}}
\newcommand{\iu}[0]{\I}
\newcommand{\defequal}[0]{\mathrel{\mathop{:}}=}
\newcommand{\di}[1]{\;\D#1}
\newcommand{\spacedot}[0]{\,\cdot\,}

\DeclareMathOperator{\sinc}{sinc}
\DeclareMathOperator*{\esssup}{ess\,sup}

\begin{document}

\title*{Signal and System Approximation from General Measurements}
\subtitle{\mdseries \normalsize Dedicated to Professor Paul Butzer on his 85th birthday}
\author{Holger~Boche\thanks{H. Boche was supported by the German Research Foundation (DFG)
    under grant \mbox{BO 1734/13-2}.} and Ullrich~J.~M\"onich\thanks{U. M{\"o}nich was supported
    by the German Research Foundation (DFG) under grant \mbox{MO 2572/1-1}.}}
\institute{Holger~Boche \at Technische Universit\"at M\"unchen,
Lehrstuhl f\"ur Theoretische Informationstechnik\\
Arcisstr. 21, 80290 M\"unchen, Germany, \email{boche@tum.de}
\and Ullrich~J.~M\"onich \at Massachusetts Institute of Technology,
Research Laboratory of Electronics\\
77 Massachusetts Avenue, Cambridge, MA 02139, USA, \email{moenich@mit.edu}}
\maketitle

\abstract{In this paper we analyze the behavior of system approximation processes for
  stable linear time-invariant (LTI) systems and signals in the Paley--Wiener space
  $\PWs_{\pi}^{1}$. We consider approximation processes, where the input signal is not
  directly used to generate the system output, but instead a sequence of numbers is used
  that is generated from the input signal by measurement functionals. We consider
  classical sampling which corresponds to a pointwise evaluation of the signal, as well as
  several more general measurement functionals. We show that a stable system approximation
  is not possible for pointwise sampling, because there exist signals and systems such
  that the approximation process diverges. This remains true even with oversampling.
  However, if more general measurement functionals are considered, a stable approximation
  is possible if oversampling is used. Further, we show that without oversampling we have
  divergence for a large class of practically relevant measurement procedures.}

\section{Introduction}

Sampling theory plays a fundamental role in modern signal and information processing,
because it is the basis for today's digital world \cite{shannon49}. The reconstruction of
continuous-time signals from their samples is also essential for other applications and
theoretical concepts \cite{jerri77,higgins85,marvasti01_book}. The reconstruction of
non-bandlimited signals, which was analyzed for example in
\cite{butzer80a,butzer92,ferreira95}, will not be considered in this paper, instead we
focus on bandlimited signals. For an overview of existing sampling theorems see for
example \cite{jerri77,higgins96_book}, and \cite{butzer88}.

The core task of digital signal processing is to process data. This means that, usually,
the interest is not in a reconstruction of the sampled signal itself, but in some
processed version of it. This might be the derivative, the Hilbert transform or the output
of any other stable linear system $T$. Then the goal is to approximate the desired
transform $T f$ of a signal $f$ by an approximation process, which uses only finitely
many, not necessarily equidistant, samples of the signal $f$. Exactly as in the case of
signal reconstruction, the convergence and approximation behavior is important for
practical applications \cite{butzer12}.

Since sampling theory is so fundamental for applications it is essential to have this
theory developed rigorously. From the first beginnings in engineering, see for example
\cite{butzer11,butzer11a} for historical comments, one main goal in research was to extend
the theory to different practically relevant classes of signals and systems. The first
author's interest for the topic was aroused in discussions with Paul Butzer in the early
1990s at RWTH Aachen. Since 2005 both authors have done research in this field and
contributed with publications, see for example the second author's thesis
\cite{monich11_phdthesis} for a summary.

In order to continue the ``digital revolution'', enormous capital expenditures and
resources are used to maintain the pace of performance increase, which is described by
Moore's law. But also the operation of current communication systems requires huge amounts
of resources, e.g. energy. It is reasonable to ask whether this is necessary. In this
context, from a signal theoretic perspective, three interesting questions are: Do there
exist fundamental limits that determine which signals and systems can be implemented
digitally? In what technology---analog, digital, or mixed signal--- can the systems be
implemented? What are the necessary resources in terms of energy and hardware to implement
the systems?

Such an implementation theory is of high practical relevance, and it already influences
the system design, although there is no general system theoretic approach available yet to
answer the posed questions. For example, the question whether to use a system
implementation based on the Shannon series operating at Nyquist rate or to use an approach
based on oversampling, which comes with higher technological effort, plays a central role
in the design of modern information processing systems. A further important question
concerns the measurement procedures. Can we use classical sampling-based measurement
procedures, where the signal values are taken at certain time instants, or is it better to
use more general measurement procedures? As already mentioned, no general methodical
approach is known that could answer these questions. Regardless of these difficulties,
Hilbert's vision applies: ``We must know. We will know.''

In this paper we analyze the convergence behavior of system approximation processes for
different kinds of sampling procedures. The structure of this paper is as follows: First,
we introduce some notation in Section~\ref{1}. Then, we treat pointwise
sampling in Section~\ref{4}. In Section~\ref{40} we
study general sampling functionals and oversampling. In Section~\ref{65} we
analyze the convergence of subsequences of the approximation process. Finally, in
Section~\ref{83} we discuss the structure of more general measurement
functionals.

The material in this paper will be presented in part at the IEEE International Conference
on Acoustics, Speech, and Signal Processing 2014 (ICASSP 2014)
\mbox{\cite{boche14a_submitted,boche14b_submitted}}.

\section{Notation}\label{1}

In order to continue the discussion, we need some preliminaries and notation. Let $\Ft{f}$
denote the Fourier transform of a function $f$, where $\Ft{f}$ is to be understood in the
distributional sense. By $L^p(\R)$, $1\leq p \leq \infty$, we denote the usual
$L^p$-spaces with norm $\lVert\spacedot\rVert_p$. $C[a,b]$ is the space of all continuous
functions on $[a,b]$. Further, $l^p$, $1\leq p < \infty$, is the space of all sequences
that are summable to the $p$th power.

For $\sigma > 0$ let $\Bs_{\sigma}$ be the set of all entire functions $f$ with the
property that for all $\epsilon > 0$ there exists a constant $C(\epsilon)$ with $\lvert
f(z) \rvert \leq C(\epsilon) \exp \bigl( (\sigma+\epsilon) \lvert z \rvert \bigr)$ for all
$z \in \C$. The Bernstein space $\Bs_{\sigma}^{p}$ consists of all functions in
$\Bs_{\sigma}$, whose restriction to the real line is in $L^p(\R)$, $1 \leq p \leq
\infty$. A function in $\Bs_{\sigma}^{p}$ is called bandlimited to $\sigma$. By the
Paley--Wiener--Schwartz theorem, the Fourier transform of a function bandlimited to
$\sigma$ is supported in $[-\sigma,\sigma]$. For $1\leq p \leq 2$ the Fourier
transformation is defined in the classical and for $p>2$ in the distributional sense. It
is well known that $\Bs_{\sigma}^{p} \subset \Bs_{\sigma}^{s}$ for $1 \leq p \leq s \leq
\infty$. Hence, every function $f\in \Bs_{\sigma}^{p}$, $1 \leq p \leq \infty$, is
bounded.

For $-\infty <\sigma_1 < \sigma_2<\infty$ and $1\leq p\leq \infty$ we denote by
$\PWs_{[\sigma_1,\sigma_2]}^{p}$ the Paley--Wiener space of functions $f$ with a
representation $f(z)=1/(2\pi) \int_{-\sigma_1}^{\sigma_2} g(\omega) \e^{\iu z\omega}
\di{\omega}$, $z \in \C$, for some $g \in L^p[\sigma_1,\sigma_2]$. The norm for
$\PWs_{[\sigma_1,\sigma_2]}^{p}$, $1\leq p < \infty$, is given by $\lVert f
\rVert_{\PWs_{[\sigma_1,\sigma_2]}^{p}}=( 1/(2\pi)\int_{\sigma_1}^{\sigma_2} \lvert
\Ft{f}(\omega) \rvert^p \di{\omega} )^{1/p}$. For $\PWs_{[-\sigma,\sigma]}^{p}$,
$0<\sigma<\infty$, we use the abbreviation $\PWs_{\sigma}^{p}$. The nomenclature
concerning the Bernstein and Paley--Wiener spaces, we introduced so far, is not consistent
in the literature. Sometimes the space that we call Bernstein space is called
Paley--Wiener space \cite{seip98}. We adhere to the notation used in
\cite{higgins96_book}.

Since our analyses involve stable linear time-invariant (LTI) systems, we briefly review
some definitions and facts. A linear system $T:\PWs_{\pi}^{p} \rightarrow \PWs_{\pi}^{p}$,
$1 \leq p \leq \infty$, is called stable if the operator $T$ is bounded, i.e., if $\lVert
T \rVert \defequal \sup_{\lVert f \rVert_{\PWs_{\pi}^{p}}\leq 1} \lVert Tf
\rVert_{\PWs_{\pi}^{p}} < \infty$. Furthermore, it is called time-invariant if
$(Tf(\spacedot - a))(t)=(Tf)(t-a)$ for all $f \in \PWs_{\pi}^{p}$ and $t,a \in \R$.

For every stable LTI system $T:\PWs_{\pi}^{1} \rightarrow \PWs_{\pi}^{1}$ there exists
exactly one function $\Ft{h}_T \in L^\infty[-\pi,\pi]$ such that
\begin{equation}\label{2}
  (Tf)(t)=\frac{1}{2\pi} \int_{-\pi}^{\pi} 
  \Ft{f}(\omega) \Ft{h}_T(\omega) \e^{\iu \omega t} \di{\omega} ,\quad t \in \R,
\end{equation}
for all $f \in \PWs_{\pi}^{1}$ \cite{boche10g}. Conversely, every function $\Ft{h}_T \in
L^\infty[-\pi,\pi]$ defines a stable LTI system $T:\PWs_{\pi}^{1} \rightarrow
\PWs_{\pi}^{1}$. The operator norm of a stable LTI system $T$ is given by $\lVert T
\rVert=\lVert \Ft{h} \rVert_{L^\infty[-\pi,\pi]}$. Furthermore, it can be shown that the
representation \eqref{2} with $\Ft{h}_T \in
L^\infty[-\pi,\pi]$ is also valid for all stable LTI systems $T: \PWs_{\pi}^{2}
\rightarrow \PWs_{\pi}^{2}$. Therefore, every stable LTI system that maps $\PWs_{\pi}^{1}$
in $\PWs_{\pi}^{1}$ maps $\PWs_{\pi}^{2}$ in $\PWs_{\pi}^{2}$, and vice versa. Note that
$\Ft{h}_T \in L^\infty[-\pi,\pi] \subset L^2[-\pi,\pi]$, and consequently $h_T \in
\PWs_{\pi}^{2}$.

An LTI system can have different representations. In textbooks, usually the frequency
domain representation \eqref{2}, and the time domain
representation in the form of a convolution integral
\begin{equation}\label{3}
  (Tf)(t)=\int_{-\infty}^{\infty} f(\tau) h_T(t-\tau) \di{\tau} 
\end{equation}
are given \cite{franks69_book,oppenheim09_book}. Although both are well-defined for stable
LTI systems $T:\PWs_{\pi}^{2} \rightarrow \PWs_{\pi}^{2}$ operating on $\PWs_{\pi}^{2}$,
there are systems and signal spaces where these representations are meaningless, because
they are divergent \cite{ferreira99,boche08c}. For example, it has been shown that there
exist stable LTI systems $T:\PWs_{\pi}^{1} \rightarrow \PWs_{\pi}^{1}$ that do not have a
convolution integral representation in the form of
\eqref{3}, because the integral diverges for certain
signals $f \in \PWs_{\pi}^{1}$ \cite{boche08c}. However, the frequency domain
representation \eqref{2}, which we will use in this
paper, holds for all stable LTI systems $T:\PWs_{\pi}^{1} \rightarrow \PWs_{\pi}^{1}$.

\section{Sampling-Based Measurements}\label{4}

\subsection{Basics of Non-Equidistant Sampling}

In the classical non-equidistant sampling setting the goal is to reconstruct a bandlimited
signal $f$ from its non-equidistant samples $\{f(t_k)\}_{k \in \Z}$, where $\{t_k\}_{k \in
  \Z}$ is the sequence of sampling points. One possibility to do the reconstruction is to
use the sampling series
\begin{equation}\label{5}
  \sum_{k=-\infty}^{\infty} f(t_k) \phi_k(t) ,
\end{equation}
where the $\phi_k$, $k \in \Z$, are certain reconstruction functions.

In this paper we restrict ourselves to sampling point sequences $\{t_k\}_{k \in \Z}$ that
are real and a complete interpolating sequence for $\PWs_{\pi}^{2}$.
\begin{definition}\label{6}
  We say that $\{t_k\}_{k\in \Z}$ is a complete interpolating sequence
  for $\PWs_{\pi}^{2}$
  if the interpolation problem $f(t_k) = c_k$, $k\in \Z$, has exactly one
  solution $f \in \PWs_{\pi}^{2}$ for every sequence $\{c_k\}_{k \in
    \Z}\in l^2$.
\end{definition}
We further assume that the sequence of sampling points $\{t_k\}_{k \in \Z}$ is ordered
strictly increasingly, and, without loss of generality, we assume that $t_0=0$. Then, it
follows that the product
\begin{equation}\label{7}
  \phi(z)=z \lim_{N \rightarrow \infty} \prod_{\substack{\lvert k \rvert \leq N\\k
    \neq 0}}
    \left( 1-\frac{z}{t_k} \right)
\end{equation}
converges uniformly on $\lvert z \rvert \leq R$ for all $R<\infty$, and $\phi$ is an
entire function of exponential type $\pi$
\cite{levin96_book}.
It can be seen from \eqref{7} that $\phi$, which is often called
generating function, has the zeros $\{t_k\}_{k \in \Z}$.
Moreover, it follows that
\begin{equation}\label{8}
  \phi_k(t)=\frac{\phi(t)}{\phi'(t_k)(t-t_k)} 
\end{equation}
is the unique function in $\PWs_{\pi}^{2}$ that solves the
interpolation problem $\phi_k(t_l)=\delta_{kl}$, where $\delta_{kl}=1$
if $k=l$, and $\delta_{kl}=0$ otherwise.

\begin{definition}\label{9}
  A system of vectors $\{\phi_k\}_{k \in \Z}$ in a separable Hilbert
  space $\mathcal{H}$ is called Riesz basis if $\{\phi_k\}_{k \in \Z}$ is
  complete in $\mathcal{H}$, and there exist positive constants $A$
  and $B$ such that for all $M,N \in \N$ and arbitrary scalars $c_k$
  we have
  \begin{equation}\label{10}
  A \sum_{k=-M}^{N} |c_k|^2
  \leq
  \left\| \sum_{k=-M}^{N} c_k \, \phi_k \right\|^2
  \leq
  B \sum_{k=-M}^{N} |c_k|^2 .
\end{equation}
\end{definition}

A well-known fact is the following theorem \cite[p. 143]{young01_book}.
\begin{theorem}[Pavlov]\label{11}
  The system $\{\e^{\iu \omega t_k} \}_{k \in \Z}$ is a Riesz basis for $L^2[-\pi,\pi]$ if
  and only if $\{t_k\}_{k \in \Z}$ is a complete interpolating sequence for
  $\PWs_{\pi}^{2}$.
\end{theorem}
It follows immediately from Theorem~\ref{11} that $\{\phi_k\}_{k \in \Z}$, as
defined in \eqref{8}, is a Riesz basis for $\PWs_{\pi}^{2}$ if $\{t_k\}_{k \in
  \Z}$ is a complete interpolating sequence for $\PWs_{\pi}^{2}$.

For further results and background information on non-equidistant sampling we would like
to refer the reader to \cite{higgins96_book,marvasti01_book}.

\subsection{Basics of Sampling-Based System Approximation}\label{12}

In many signal processing applications the goal is to process a signal $f$. In this paper
we consider signals from the space $\PWs_{\pi}^{1}$. A common method to do such a
processing is to use LTI systems. Given a signal $f \in \PWs_{\pi}^{1}$ and a stable LTI
system $T \colon \PWs_{\pi}^{1} \to \PWs_{\pi}^{1}$ we can use
\eqref{2} to calculate the desired system output $Tf$.
Equation \eqref{2} can be seen as an analog
implementation of the system $T$. As described in Section~\ref{1},
\eqref{2} is well defined for all $f \in \PWs_{\pi}^{1}$
and all stable LTI systems $T \colon \PWs_{\pi}^{1} \to \PWs_{\pi}^{1}$, and we have no
convergence problems.

However, often only the samples $\{f(t_k)\}_{k \in \Z}$ of a signal are available, like it
is the case in digital signal processing, and not the whole signal. In this situation we
seek an implementation of the stable LTI system $T$ which uses only the samples
$\{f(t_k)\}_{k \in \Z}$ of the signal $f$ \cite{stens83}. We call such an implementation
an implementation in the digital domain. For example, the sampling series
\begin{equation}\label{13}
  \sum_{k=-\infty}^{\infty} f(t_k) (T \phi_k)(t) 
\end{equation}
is a digital implementation of the system $T$.
However, in contrast to \eqref{2}, the convergence of
\eqref{13} is not guaranteed, as we will see in
Section~\ref{19}.

In Figure~\ref{14} the different approaches that are taken for an
analog and a digital system implementation are visualized. The general motive for the
development of the ``digital world'' is the idea that every stable analog system can be
implemented digitally, i.e., that the diagram in Figure~\ref{14} is
commutative.

\begin{figure}
  \centering
  \includegraphics{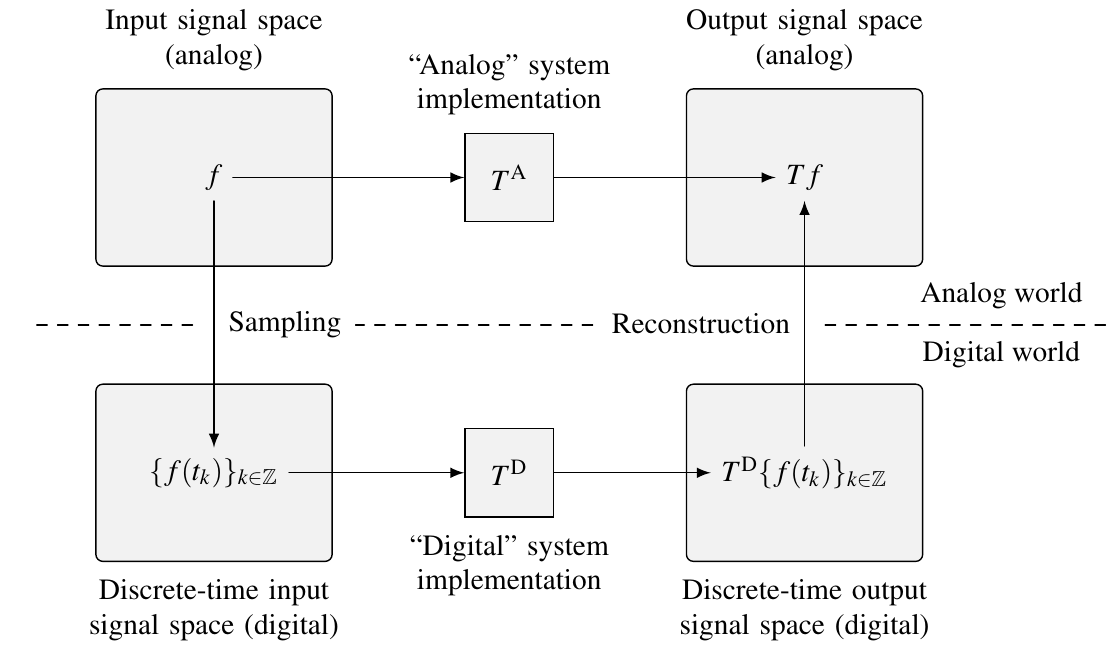}
  \caption{Analog versus digital system implementation of a stable LTI system $T$.}
  \label{14}
\end{figure}

\begin{remark}\label{15}
  In this paper the systems are always linear and well defined. However, there exist practically
  important systems that do not exist as a linear system \cite{boche10f}.
  For a discussion about non-linear systems, see \cite{ferreira01}.
\end{remark}

\subsection{Two Conjectures}

In \cite{boche11a} we posed two conjectures, which we will prove in this paper. The first
conjecture is about the divergence of the system approximation process for complete
interpolating sequences in the case of classical pointwise sampling.
\begin{conjecture}\label{16}
  Let $\{t_k\}_{k\in \Z} \subset \R$ be an ordered complete interpolating sequence for
  $\PWs_{\pi}^{2}$, $\phi_k$ as defined in \eqref{8}, and $0<\sigma<\pi$. Then,
  for all $t \in \R$ there exists a stable LTI system $T_* \colon \PWs_{\pi}^{1} \to
  \PWs_{\pi}^{1}$ and a signal $f_* \in \PWs_{\sigma}^{1}$ such that
  \begin{equation*}
    \limsup_{N \rightarrow \infty} \left| (T_*f_*)(t) - \sum_{k=-N}^{N}
      f_*(t_k) (T_*\phi_k)(t) \right|
    =
    \infty .
  \end{equation*}
\end{conjecture}

For the special case of equidistant sampling, the system approximation process
\eqref{13} reduces to
\begin{equation}\label{17}
  \frac{1}{a} \sum_{k=-\infty}^{\infty} f \left( \frac{k}{a} \right) h_T \left(
    t-\frac{k}{a} \right) ,
\end{equation}
where $a \geq 1$ denotes the oversampling factor and $h_T$ is the impulse response of the
system $T$. It has already been shown that the Hilbert transform is a universal system for
which there exists, for every amount of oversampling, a signal such that the peak value of
\eqref{17} diverges \cite{boche10g}. In
Conjecture~\ref{16} now, the statement is that this divergence even occurs for
non-equidistant sampling, which introduces an additional degree of freedom, and even
pointwise. However, in this case, the Hilbert transform is no longer the universal
divergence creating system.

Conjecture~\ref{16} will be proved in Section~\ref{19}.

The second conjecture is about more general measurement procedures and states that with
suitable measurement procedures and oversampling we can obtain a convergent approximation
process.
\begin{conjecture}\label{18}
  Let $\{t_k\}_{k\in \Z} \subset \R$ be an ordered complete interpolating sequence for
  $\PWs_{\pi}^{2}$, $\phi_k$ as defined in \eqref{8}, and $0<\sigma<\pi$. There
  exists a sequence of continuous linear functionals $\{c_k\}_{k \in \Z}$ on
  $\PWs_{\pi}^{1}$ such that for all stable LTI systems $T \colon \PWs_{\pi}^{1} \to
  \PWs_{\pi}^{1}$ and all $f \in \PWs_{\sigma}^{1}$ we have
  \begin{equation*}
    \lim_{N \rightarrow \infty} \sup_{t \in \R} \left| (Tf)(t)-\sum_{k=-N}^{N} c_k(f) \, (T
      \phi_k)(t) \right|
    =
    0.
  \end{equation*}  
\end{conjecture}
Conjecture~\ref{18} will be proved in Section~\ref{40}, where we
also introduce the general measurement procedures more precisely.

\subsection{Approximation for Sampling-Based Measurements}\label{19}

In this section we analyze the system approximation process which is given by the digital
implementation \eqref{13}. The next theorem proves
Conjecture~\ref{16}.

\begin{theorem}\label{20}
  Let $\{t_k\}_{k \in \Z} \subset \R$ be an ordered complete interpolating sequence for
  $\PWs_{\pi}^{2}$, $\phi_k$ as defined in \eqref{8}, and $t \in \R$. Then
  there exists a stable LTI system $T_* \colon \PWs_{\pi}^{1} \to \PWs_{\pi}^{1}$ such
  that for every $0<\sigma<\pi$ there exists a signal $f_* \in \PWs_{\sigma}^{1}$ such
  that
  \begin{equation}\label{21}
    \limsup_{N \rightarrow \infty} \left| \sum_{k=-N}^{N} f_*(t_k) (T_* \phi_k)(t) \right|
    =
    \infty .
  \end{equation}
\end{theorem}

\begin{remark}
  It is interesting to note that the system $T_*$ in
  Theorem~\ref{20} is universal in the sense that it does
  not depend on $\sigma$, i.e., on the amount of oversampling. In other words, we can find
  a stable LTI system $T_*$ such that regardless of the oversampling factor
  $1<\alpha<\infty$ there exists a signal $f_* \in \PWs_{\pi/\alpha}^{1}$ for which the
  system approximation process diverges as in
  \eqref{21}.
\end{remark}

\begin{remark}\label{22}
  Since $\{\phi_k\}_{k \in \Z}$ is a Riesz basis for $\PWs_{\pi}^{2}$, it follows that the
  projections of $\{\phi_k\}_{k \in \Z}$ onto $\PWs_{\sigma}^{2}$ form a frame for
  $\PWs_{\sigma}^{2}$, $0<\sigma<\pi$ \cite[p.~231]{heil11_book}.
  Theorem~\ref{20} shows that the usually nice behavior of
  frames is destroyed in the presence of a system $T$. Even though the projections of
  $\{\phi_k\}_{k \in \Z}$ onto $\PWs_{\sigma}^{2}$ form a frame for $\PWs_{\sigma}^{2}$,
  $0<\sigma<\pi$, we have divergence when we add the system $T$. This behavior was known
  before for pointwise sampling: The reconstruction functions in the Shannon sampling
  series form a Riesz basis for $\PWs_{\pi}^{2}$, and the convergence of the series is
  globally uniform for signals in $\PWs_{\sigma}^{1}$, $0<\sigma<\pi$, i.e., if
  oversampling is applied. However, with a system $T$ we can have even pointwise
  divergence \cite{boche10g}. Theorem~\ref{20} illustrates
  that this is true not only for pointwise sampling but also if more general measurement
  functionals are used.
\end{remark}

\begin{remark}
  The system $T_*$ from Theorem~\ref{20} can, as a stable
  LTI system, of course be implemented, using the analog system implementation
  \eqref{2}. However,
  Theorem~\ref{20} shows that a digital, i.e., sampling
  based, implementation is not possible. This also illustrates the limits of a general
  sampling-based technology. We will see later, in
  Section~\ref{46}, that the system can be
  implemented by using more general measurement functionals and oversampling.
\end{remark}

The result of Theorem~\ref{20} is also true for bandpass
signals. However, in this case the stable LTI system $T_*$ is no longer universal but
depends on the actual frequency support of the signal space.

\begin{theorem}\label{23}
  Let $\{t_k\}_{k \in \Z} \subset \R$ be an ordered complete interpolating sequence for
  $\PWs_{\pi}^{2}$, $\phi_k$ as defined in \eqref{8}, $t \in \R$, and
  $0<\sigma_1<\sigma_2<\pi$. Then there exist a stable LTI system $T_* \colon
  \PWs_{\pi}^{1} \to \PWs_{\pi}^{1}$ and a signal $f_* \in \PWs_{[\sigma_1,\sigma_2]}^{1}$
  such that
  \begin{equation*}
    \limsup_{N \rightarrow \infty} \left| \sum_{k=-N}^{N} f_*(t_k) (T_* \phi_k)(t) \right|
    =
    \infty .
  \end{equation*}
\end{theorem}
For the proof of Theorems \ref{20} and
\ref{23}, we need two lemmas,
Lemma~\ref{24} and Lemma~\ref{37}. The proof of
Lemma~\ref{24} heavily relies on a result of Szarek, which was published in
\cite{szarek80}.

\begin{lemma}\label{24}
  Let $\{t_k\}_{k \in \Z} \subset \R$ be an ordered complete interpolating sequence for
  $\PWs_{\pi}^{2}$ and $\phi_k$ as defined in \eqref{8}. Then there exists a
  positive constant $C_{1}$ such that for all $\omega \in [-\pi,\pi]$ and all $N \in \N$
  we have
  \begin{equation}\label{25}
    \max_{1 \leq M \leq N}
    \frac{1}{2\pi} \int_{-\pi}^{\pi} \left| \sum_{k=-M}^{M} \e^{\iu \omega t_k}
      \Ft{\phi}_k(\omega_1) \right| \di{\omega_1}
    \geq
    C_{1} \log(N) .
  \end{equation}
\end{lemma}
\begin{remark}\label{26}
  Later, in Section~\ref{65}, we will see what potential implications the
  presence of the $\max$-operator in \eqref{25} can have on the
  convergence behavior of the approximation process. Currently, our proof technique is not
  able to show more, however, we conjecture that \eqref{25} is also true
  without $\max_{1 \leq M \leq N}$.
\end{remark}
For the proof of Lemma~\ref{24} we need Lemmas 2 and 3 from Szarek's paper
\cite{szarek80}. For completeness and convenience, we state them next in a slightly
simplified version, which is sufficient for our purposes.
\begin{lemma}[Szarek]\label{27}
  Let $f$ be a nonnegative measurable function, $C_{2}$ a positive constant, and $n$ a
  natural number such that
  \begin{equation}\label{28}
    \frac{1}{2\pi} \int_{-\pi}^{\pi} (f(t))^2 \di{t}
    \leq
    C_{2} n
  \end{equation}
  and
  \begin{equation}\label{29}
    \frac{1}{2\pi} \int_{-\pi}^{\pi} (f(t))^{5/4} \di{t}
    \geq
    \frac{n^{1/4}}{ C_{2}}.
  \end{equation}
  Then there exists a number $\alpha=\alpha( C_{2})$, $0<\alpha<2^{-3}$ and
  a natural number $s$ such that
  \begin{equation*}
    \frac{1}{2\pi} \int_{\{t \in [-\pi,\pi] : f(t) > \frac{n}{\alpha^2}\}} f(t) \di{t}
    \leq
    \frac{\alpha}{2^4}
  \end{equation*}
  and
  \begin{equation*}
    \frac{1}{2\pi} \int_{\{t \in [-\pi,\pi] : \frac{\alpha^s n}{\alpha^2} < f(t) \leq
      \frac{\alpha^s n}{\alpha^3}\}} f(t) \di{t}
    \geq
    s \alpha .
  \end{equation*}
\end{lemma}

\begin{lemma}[Szarek]\label{30}
  Let $0<\alpha<2^{-3}$ and $\{F_k\}_{k=1}^{N}$ be a sequence of measurable functions.
  Further, define $F_{k,n} \defequal F_{k+n} - F_{k}$. Assume that for all $k,n$
  satisfying $1 \leq k,n$ and $1 \leq k+n \leq N$ there exists a natural number $s=s(k,n)$
  such that
  \begin{equation*}
    \frac{1}{2\pi} \int_{\{t \in [-\pi,\pi] : \lvert F_{k,n}(t) \rvert >
      \frac{n}{\alpha^2}\}} \lvert F_{k,n}(t) \rvert \di{t}
    \leq
    \frac{\alpha}{2^4}   
  \end{equation*}
  and
  \begin{equation*}
    \frac{1}{2\pi} \int_{\{t \in [-\pi,\pi] : \frac{\alpha^s n}{\alpha^2} < \lvert
      F_{k,n}(t) \rvert \leq \frac{\alpha^s n}{\alpha^3}\}} \lvert F_{k,n}(t) \rvert
    \di{t}
    \geq
    s \alpha .
  \end{equation*}
  Then there exists a positive constant
  $C_{3}=C_{3}(\alpha)$ such that
  \begin{equation*}
    \max_{1 \leq k \leq N} \frac{1}{2\pi} \int_{-\pi}^{\pi} \lvert F_k(t) \rvert \di{t}
    \geq
    C_{3}(\alpha) \log(N) .
  \end{equation*}
\end{lemma}

Now we are in the position to prove Lemma~\ref{24}.

\begin{proof}[Lemma~\ref{24}]
  \smartqed Let $\{t_k\}_{k \in \Z} \subset \R$ be an arbitrary but fixed ordered complete
  interpolating sequence for $\PWs_{\pi}^{2}$ and $\phi_k$ as defined in
  \eqref{8}. Further, let $\omega \in [-\pi,\pi]$ be arbitrary but fixed. For
  $\omega_1 \in [-\pi,\pi]$ consider the functions
  \begin{equation*}
    G_k(\omega_1,\omega)
    \defequal
    \sum_{l=-k}^{k} \e^{\iu \omega t_l} \Ft{\phi}_l(\omega_1) ,
  \end{equation*}
  and
  \begin{align*}
    G_{k,n}(\omega_1,\omega)
    &\defequal
    G_{k+n}(\omega_1,\omega)
    -
    G_{k}(\omega_1,\omega) \\
    &=
    \sum_{k<\lvert l \rvert \leq k+n} \e^{\iu \omega t_l} \Ft{\phi}_l(\omega_1) .
  \end{align*}
  We will show that $\lvert G_{k,n}(\omega_1,\omega) \rvert$ satisfies the conditions
  \eqref{28} and \eqref{29} of Lemma~\ref{27}.

  We have
  \begin{align}
    \frac{1}{2\pi}
    \int_{-\pi}^{\pi} \lvert G_{k,n}(\omega_1,\omega) \rvert^2 \di{\omega_1}
    &=
    \int_{-\infty}^{\infty} \left| \sum_{k<\lvert l \rvert \leq k+n} \e^{\iu \omega t_l}
      \phi_l(t) \right|^2 \di{t} \notag \\
    &\leq
    B \sum_{k<\lvert l \rvert \leq k+n} 1 \notag \\
    &=
    B 2n , \label{31}
  \end{align}
  where we used the fact that $\{\phi_l\}_{k \in \Z}$ is a Riesz basis for $\PWs_{\pi}^{2}$.
  
  Next, we analyze the expression
  \begin{equation*}
    \frac{1}{2\pi} \int_{-\pi}^{\pi}  \left| G_k(\omega_1,\omega) \right|^p \di{\omega_1}
    =
    \frac{1}{2\pi} \int_{-\pi}^{\pi}  \left| \sum_{l=-k}^{k} \e^{\iu \omega t_l}
      \Ft{\phi}_l(\omega_1) \right|^p \di{\omega_1}
  \end{equation*}
  for $1<p<2$.
  We set
  \begin{equation}\label{32}
    G_{k,n}(\omega_1,\omega)
    =
    0 \qquad \text{for} \quad \text{$\lvert \omega_1 \rvert > \pi$}
  \end{equation}
  and consider the Fourier transform
  \begin{equation}\label{33}
    (\Fto G_{k,n}(\spacedot,\omega))(t)
    =
    \int_{-\infty}^{\infty} G_{k,n}(\omega_1,\omega) \e^{-\iu \omega_1 t} \di{\omega_1} .
  \end{equation}
  Due to \eqref{32}, the integral in \eqref{33} is absolutely
  convergent.
  We have
  \begin{equation*}
    (\Fto G_{k,n}(\spacedot,\omega))(t)
    =
    2\pi g_{k,n}(-t,\omega) ,
  \end{equation*}
  where
  \begin{equation*}
    g_{k,n}(t,\omega)
    \defequal
    \sum_{k<\lvert l \rvert \leq k+n} \e^{\iu \omega t_l} \phi_l(t) .
  \end{equation*}
  Let $q$ be the conjugate of $p$, i.e., $1/p+1/q=1$, then the Hausdorff--Young
  inequality \cite{butzer94},\cite[p. 19]{higgins96_book} shows that there exists a constant
  $C_{4}=C_{4}(p)$ such that
  \begin{equation*}
    \left( \int_{-\infty}^{\infty} \lvert (\Fto G_{k,n}(\spacedot,\omega))(t) \rvert^q
      \di{t} \right)^{\frac{1}{q}}
    \leq
    C_{4}(p) \left( \int_{-\pi}^{\pi} \lvert G_{k,n}(\omega_1,\omega) \rvert^p
      \di{\omega_1} \right)^{\frac{1}{p}} ,
  \end{equation*}
  which implies that
  \begin{align}\label{34}
    \frac{1}{2\pi}
    \int_{-\pi}^{\pi} \lvert G_{k,n}(\omega_1,\omega) \rvert^p
    \di{\omega_1}
    \geq
    \frac{(2 \pi)^{p-1}}{( C_{4}(p))^p}
    \left( \int_{-\infty}^{\infty} \lvert g_{k,n}(t,\omega) \rvert^q \di{t}
    \right)^{\frac{p}{q}} .
  \end{align}
  Note that the constant $C_{4}(p)$ is independent of $\omega$.
  We analyze the integral on the right-hand side of \eqref{34}.
  We have $g_{k,n}(\spacedot,\omega) \in \Bs_{\pi}^{q}$.
  Since $\{t_k\}_{k \in \Z}$ is a complete interpolating sequence for $\PWs_{\pi}^{2}$, we
  have \cite{pavlov79}
  \begin{equation*}
    \inf_{k \in \Z} (t_{k+1}-t_{k})
    >
    0 ,
  \end{equation*}
  and it is known \cite[p.~101]{boas54_book} that there exists a
  positive constant $C_{5}(q)$ that is independent of
  $k$, $n$, and $\omega$ such that
  \begin{equation*}
    \left( \int_{-\infty}^{\infty} \lvert g_{k,n}(t,\omega) \rvert^q \di{t}
    \right)^{\frac{1}{q}}
    \geq
    C_{5}(q)
    \left( \sum_{l=-\infty}^{\infty} \lvert g_{k,n}(t_l,\omega) \rvert^q 
    \right)^{\frac{1}{q}} .
  \end{equation*}
  Since
  \begin{equation*}
    \sum_{l=-\infty}^{\infty} \lvert g_{k,n}(t_l,\omega) \rvert^q
    =
    \sum_{k<\lvert l \rvert \leq n+k} 1
    =
    2n ,
  \end{equation*}
  we obtain
  \begin{equation}\label{35}
    \left( \int_{-\infty}^{\infty} \lvert g_{k,n}(t,\omega) \rvert^q \di{t}
    \right)^{\frac{1}{q}}
    \geq
    C_{5}(q) (2n)^{\frac{1}{q}} .
  \end{equation}
  Combining \eqref{34} and \eqref{35} gives
  \begin{align*}
    \frac{1}{2\pi}
    \int_{-\pi}^{\pi} \lvert G_{k,n}(\omega_1,\omega) \rvert^p
    \di{\omega_1}
    \geq
    \frac{(2 \pi)^{p-1}( C_{5}(q))^p}{( C_{4}(p))^p}
    (2n)^{\frac{p}{q}} ,
  \end{align*}
  and for $p=5/4$ we obtain
  \begin{equation}\label{36}
    \frac{1}{2\pi}
    \int_{-\pi}^{\pi} \lvert G_{k,n}(\omega_1,\omega) \rvert^{\frac{5}{4}}
    \di{\omega_1}
    \geq
    \frac{(4 \pi)^{\frac{1}{4}} ( C_{5}(5))^{\frac{5}{4}}}{( C_{4}(\frac{5}{4}))^{\frac{5}{4}}}
    n^{\frac{1}{4}} .
  \end{equation}
  
  Choosing
  \begin{equation*}
    C_{2}=\max \left\{2 B^2, \frac{( C_{4}(\frac{5}{4}))^{\frac{5}{4}}}{(4
        \pi)^{\frac{1}{4}} ( C_{5}(5))^{\frac{5}{4}}} \right\} ,
  \end{equation*}
  we see from \eqref{31} and \eqref{36} that the
  function $\lvert G_{k,n}(\omega_1,\omega) \rvert$ satisfies conditions
  \eqref{28} and \eqref{29}, that is the assumptions of
  Lemma~\ref{27}. Hence, as a result of Lemma~\ref{27}, $\lvert
  G_{k}(\omega_1,\omega) \rvert$ also satisfies the assumptions of Lemma~\ref{30},
  and application of Lemma~\ref{30} completes the proof. \qed
\end{proof}

Next, we state the second lemma which we need for the proofs of Theorems
\ref{20} and
\ref{23}. We will use it to analyze the influence
of the transfer function $\Ft{h}_T$ on the approximation process.
\begin{lemma}\label{37}
  Let $\{t_k\}_{k \in \Z} \subset \R$ be an ordered complete interpolating sequence for
  $\PWs_{\pi}^{2}$ and $\phi_k$ as defined in \eqref{8}. For all $\omega \in
  [-\pi,\pi]$, all $t \in \R$, and all $N \in \N$ we have
  \begin{align*}
    &\sup_{\substack{\lVert \Ft{g} \rVert_{L^\infty[-\pi,\pi]} \leq 1\\\Ft{g} \in C[-\pi,\pi]}}
    \left| \sum_{k=-N}^{N} \e^{\iu \omega t_k} \frac{1}{2\pi} \int_{-\pi}^{\pi}
      \Ft{g}(\omega_1) \Ft{\phi}_k(\omega_1) \e^{\iu \omega_1 t} \di{\omega_1} \right| \\
    &\hspace{5cm}=
    \frac{1}{2\pi}
    \int_{-\pi}^{\pi} \left| \sum_{k=-N}^{N} \e^{\iu \omega t_k} \Ft{\phi}_k(\omega_1)
    \right| \di{\omega_1} .
  \end{align*}
\end{lemma}

\begin{proof}[Lemma~\ref{37}]
  \smartqed
  Let $\{t_k\}_{k \in \Z} \subset \R$ be an ordered complete interpolating sequence for $\PWs_{\pi}^{2}$,
  $\omega \in [-\pi,\pi]$, $t \in \R$, and $N \in \N$, all be arbitrary but fixed.
  Further, let $\phi_k$ be defined as in \eqref{8}.
  For 
  \begin{equation*}
    \Ft{g}(\omega_1)
    =
    \exp \left(
      -\iu \arg \left( \e^{\iu \omega_1 t} \sum_{k=-N}^{N} \e^{\iu \omega t_k} \Ft{\phi}_k(\omega_1) \right)
    \right)
  \end{equation*}
  we have
  \begin{equation}\label{38}
    \left| \sum_{k=-N}^{N} \e^{\iu \omega t_k} \frac{1}{2\pi} \int_{-\pi}^{\pi}
      \Ft{g}(\omega_1) \Ft{\phi}_k(\omega_1) \e^{\iu \omega_1 t} \di{\omega_1} \right|
    =
    \frac{1}{2\pi}
    \int_{-\pi}^{\pi} \left| \sum_{k=-N}^{N} \e^{\iu \omega t_k} \Ft{\phi}_k(\omega_1)
    \right| \di{\omega_1} .
  \end{equation}
  Further, as a consequence of Lusin's theorem \cite[p.~56]{rudin87_book}, there
  exists a sequence of functions $\{\Ft{g}_n\}_{n \in \N}$ with $\Ft{g}_n \in C[-\pi,\pi]$ and
  $\lVert \Ft{g}_n \rVert_{L^\infty[-\pi,\pi]} \leq 1$, such that\linebreak
  \mbox{$\lim_{n \rightarrow \infty} \Ft{g}_n(\omega_1)=\Ft{g}(\omega_1)$} almost everywhere.
  It follows from Lebesgue's dominated convergence theorem and
  \eqref{38} that
  \begin{align*}
    &\lim_{n \rightarrow \infty}
    \left| \sum_{k=-N}^{N} \e^{\iu \omega t_k} \frac{1}{2\pi} \int_{-\pi}^{\pi}
      \Ft{g}_n(\omega_1) \Ft{\phi}_k(\omega_1) \e^{\iu \omega_1 t} \di{\omega_1} \right| \\
    &\hspace{3cm}=
    \left| \sum_{k=-N}^{N} \e^{\iu \omega t_k} \frac{1}{2\pi} \int_{-\pi}^{\pi}
      \Ft{g}(\omega_1) \Ft{\phi}_k(\omega_1) \e^{\iu \omega_1 t} \di{\omega_1} \right| \\
    &\hspace{3cm}=
    \frac{1}{2\pi}
    \int_{-\pi}^{\pi} \left| \sum_{k=-N}^{N} \e^{\iu \omega t_k} \Ft{\phi}_k(\omega_1)
    \right| \di{\omega_1} .
  \end{align*}
  Hence, taking the limit $n \rightarrow \infty$ on both sides of
  \begin{align*}
    &\sup_{\substack{\lVert \Ft{g} \rVert_{L^\infty[-\pi,\pi]} \leq 1\\\Ft{g} \in C[-\pi,\pi]}}
    \left| \sum_{k=-N}^{N} \e^{\iu \omega t_k} \frac{1}{2\pi} \int_{-\pi}^{\pi}
      \Ft{g}(\omega_1) \Ft{\phi}_k(\omega_1) \e^{\iu \omega_1 t} \di{\omega_1} \right| \\
    &\hspace{3cm} \geq
    \left| \sum_{k=-N}^{N} \e^{\iu \omega t_k} \frac{1}{2\pi} \int_{-\pi}^{\pi}
      \Ft{g}_n(\omega_1) \Ft{\phi}_k(\omega_1) \e^{\iu \omega_1 t} \di{\omega_1} \right|
  \end{align*}
  completes the proof.
  \qed
\end{proof}

\begin{proof}[Theorem~\ref{20}]
  \smartqed
  Let $\{t_k\}_{k \in \Z} \subset \R$ be an arbitrary but fixed ordered complete interpolating
  sequence for $\PWs_{\pi}^{2}$ and $\phi_k$ as defined in \eqref{8}. Further,
  let $t \in \R$ be arbitrary but fixed.

From Lemma~\ref{24} we see that
\begin{equation*}
  \sup_{N \in \N}
  \int_{-\pi}^{\pi} \left| \sum_{k=-N}^{N} \e^{\iu \omega t_k} \Ft{\phi}_k(\omega_1)
  \right| \di{\omega_1}
  =
  \infty
\end{equation*}
for all $\omega \in [-\pi,\pi]$.
Due to Lemma~\ref{37} this implies that
\begin{equation*}
  \sup_{N \in \N} \Biggl( 
    \sup_{\substack{\lVert \Ft{h}_T \rVert_{L^\infty[-\pi,\pi]} \leq 1\\\Ft{h}_T \in C[-\pi,\pi]}}
    \left| \sum_{k=-N}^{N} \e^{\iu \omega t_k} \frac{1}{2\pi} \int_{-\pi}^{\pi}
      \Ft{h}_T(\omega_1) \Ft{\phi}_k(\omega_1) \e^{\iu \omega_1 t} \di{\omega_1} \right|
  \Biggr)
  =
  \infty 
\end{equation*}
for all $\omega \in [-\pi,\pi]$.
Thus, according to the Banach--Steinhaus theorem \cite[p. 98]{rudin87_book}, for all
$\omega \in [-\pi,\pi]$ there exists a
function $\Ft{h}_{T_\omega} \in C[-\pi,\pi]$ such that
\begin{equation*}
  \limsup_{N \rightarrow \infty} \Biggl( 
    \left| \sum_{k=-N}^{N} \e^{\iu \omega t_k} \frac{1}{2\pi} \int_{-\pi}^{\pi}
      \Ft{h}_{T_\omega}(\omega_1) \Ft{\phi}_k(\omega_1) \e^{\iu \omega_1 t} \di{\omega_1} \right|
  \Biggr)
  =
  \infty .
\end{equation*}
Further, since $\Ft{h}_{T_\omega} \in C[-\pi,\pi] \subset L^\infty[-\pi,\pi]$, and since there is
the bijection \eqref{2} between $L^\infty[-\pi,\pi]$ and
the set of stable LTI systems $T_\omega \colon \PWs_{\pi}^{1} \to \PWs_{\pi}^{1}$, it
follows that for all $\omega \in [-\pi,\pi]$ there exists a stable LTI system $T_\omega$
such that
\begin{equation*}
  \limsup_{N \rightarrow \infty} \left| \sum_{k=-N}^{N} \e^{\iu \omega t_k} (T_\omega \phi_k)(t) \right|
  =
  \infty .
\end{equation*}
In particular, for $\omega=0$ there exists a stable LTI system $T_*=T_0$ such that
\begin{equation}\label{39}
  \limsup_{N \rightarrow \infty} \left| \sum_{k=-N}^{N} (T_* \phi_k)(t) \right|
  =
  \infty .
\end{equation}
$T_0$ is the desired stable LTI system $T_*$.

Next, let $0<\sigma<\pi$ be arbitrary but fixed.
For $f \in \PWs_{\sigma}^{1}$ and $N \in \N$ we have
\begin{equation*}
  \sum_{k=-N}^{N} f(t_k) (T_* \phi_k)(t)
  =
  \frac{1}{2\pi}
  \int_{-\sigma}^{\sigma} \Ft{f}(\omega_1) \sum_{k=-N}^{N} \e^{\iu \omega_1 t_k} (T_* \phi_k)(t)
  \di{\omega_1} .
\end{equation*}
Hence, it follows that
\begin{align*}
  \sup_{\lVert f \rVert_{\PWs_{\sigma}^{1}} \leq 1}
  \left| \sum_{k=-N}^{N} f(t_k) (T_* \phi_k)(t) \right|
  &=
  \max_{\omega_1 \in [-\sigma,\sigma]}
  \left| \sum_{k=-N}^{N} \e^{\iu \omega_1 t_k} (T_* \phi_k)(t) \right| \\
  &\geq
  \left| \sum_{k=-N}^{N} (T_* \phi_k)(t) \right| .
\end{align*}
Consequently, from \eqref{39} we obtain that
\begin{equation*}
  \limsup_{N \rightarrow \infty}
  \left(
    \sup_{\lVert f \rVert_{\PWs_{\sigma}^{1}} \leq 1}
    \left| \sum_{k=-N}^{N} f(t_k) (T_* \phi_k)(t) \right|
  \right)
  =
  \infty .
\end{equation*}
Thus, the Banach--Steinhaus theorem \cite[p. 98]{rudin87_book} implies that there exists a
signal $f_* \in \PWs_{\sigma}^{1}$ such that
\begin{equation*}
  \limsup_{N \rightarrow \infty}
  \left| \sum_{k=-N}^{N} f_*(t_k) (T_* \phi_k)(t) \right|
  =
  \infty .
\end{equation*}
This completes the proof. \qed
\end{proof}

\begin{proof}[Theorem~\ref{23}]
  \smartqed
  The proof of Theorem~\ref{23} is identical to
  the proof of Theorem~\ref{20}, except that we choose
  $\omega \in [\sigma_1,\sigma_2]$ instead of $\omega=0$. Since the divergence
  creating stable LTI system $T_*$ depends on the actual choice of $\omega$, we see that
  $T_*$ is no longer universal in the sense that it is independent of $\sigma_1$ and $\sigma_2$.
  \qed
\end{proof}

\section{General Measurement Functionals and Oversampling}\label{40}

\subsection{Basic Properties of General Measurement Functionals}

A key concept in signal processing is to process analog, i.e., continuous-time signals in
the digital domain. The fist step in this procedure is to convert the continuous-time
signal into a discrete-time signal, i.e., into a sequence of numbers. In
Section~\ref{4} we analyzed a sampling-based system approximation,
where the point evaluation functionals $f \mapsto f(t_k)$ are used to do this conversion.
Next, we will proceed to more general measurement functionals
\cite{partington96,butzer98,butzer00}.

The approximation of $Tf$ by the system approximation process
\begin{equation}\label{41}
  \sum_{k=-N}^{N} f(t_k) (T \phi_k)(t)
\end{equation}
can be seen as an approximation that uses the biorthogonal system
$\left\{ \e^{-\iu \spacedot t_k} , \Ft{\phi}_k \right\}_{k \in \Z}$.
In this setting, the sampling functionals, which define a certain measurement procedure, are given by
\begin{equation}\label{42}
  c_k(f)
  =
  f(t_k)
  =
  \frac{1}{2\pi} \int_{-\pi}^{\pi} \Ft{f}(\omega) \e^{\iu \omega t_k} \di{\omega} ,
\end{equation}
and the functions
\begin{equation*}
  \phi_k(t)
  =
  \frac{1}{2\pi} \int_{-\pi}^{\pi} \Ft{\phi}_k(\omega) \e^{\iu \omega t} \di{\omega}
\end{equation*}
serve as reconstruction functions in the approximation process
\eqref{41}.

In Theorem~\ref{20} we have seen that for $f \in
\PWs_{\pi}^{1}$ even with oversampling an approximation of $Tf$ using the process
\eqref{41} is not possible in general, because
there are signals $f \in \PWs_{\pi}^{1}$ and stable LTI systems $T$ such that
\eqref{41} diverges.

Next, we will study more general measurement procedures than
\eqref{42} in hopes of circumventing the divergence that was observed
in Theorem~\ref{20}. To this end, we consider a complete
orthonormal system $\{\Ft{\theta}_n\}_{n \in \N}$ in $L^2[-\pi,\pi]$.

For $f \in \PWs_{\pi}^{2}$ the situation is simple. The measurement functionals $c_n \colon
\PWs_{\pi}^{2} \to \C$ are given by
\begin{equation*}
  c_n(f)
  =
  \frac{1}{2\pi} \int_{-\pi}^{\pi} \Ft{f}(\omega) \cc{\Ft{\theta}_n(\omega)} \di{\omega}
  =
  \int_{-\infty}^{\infty} f(t) \cc{\theta_n(t)} \di{t} .
\end{equation*}
Further, we have
\begin{equation*}
  \lim_{N \rightarrow \infty} \frac{1}{2\pi} \int_{-\pi}^{\pi}
  \left| \Ft{f}(\omega) - \sum_{n=1}^{N} c_n(f) \Ft{\theta}_n(\omega) \right|^2 \di{\omega}
  =
  0
\end{equation*}
as well as
\begin{equation*}
  \lim_{N \rightarrow \infty} \int_{-\infty}^{\infty}
  \left| f(t) - \sum_{n=1}^{N} c_n(f) \theta_n(t) \right|^2 \di{t}
  =
  0
\end{equation*}
for all $f \in \PWs_{\pi}^{2}$.

In order that
\begin{equation}\label{43}
  c_n(f)
  =
  \frac{1}{2\pi} \int_{-\pi}^{\pi} \Ft{f}(\omega) \cc{\Ft{\theta}_n(\omega)} \di{\omega}
\end{equation}
is also a reasonable measurement procedure for $f \in
\PWs_{\pi}^{1}$, we need the functionals $c_n \colon \PWs_{\pi}^{1} \to \C$, defined by \eqref{43}, to be
continuous and uniformly bounded in $n$. Since 
\begin{equation*}
  \sup_{\lVert f \rVert_{\PWs_{\pi}^{1}} \leq 1} \lvert c_n(f) \rvert
  =
  \lVert \Ft{\theta}_n \rVert_{L^\infty[-\pi,\pi]} ,
\end{equation*}
this means we additionally have to require that the functions of the complete orthonormal
system $\{\Ft{\theta}_n\}_{n \in \N}$ satisfy 
\begin{equation}\label{44}
  \sup_{n \in \N} \lVert \Ft{\theta}_n \rVert_{L^\infty[-\pi,\pi]}
  <
  \infty .
\end{equation}

Using these more general measurement functionals \eqref{43}, the
system approximation process takes the form
\begin{equation}\label{45}
  \sum_{n=1}^{\infty} c_n(f) (T \theta_n)(t) .
\end{equation}
In the next section we study the approximation process
\eqref{45} and analyze its convergence behavior for
signals $f \in \PWs_{\sigma}^{1}$, $0<\sigma<\pi$. We will see that with these more
general linear measurement functionals a stable implementation of LTI systems is possible.

A special case of measurement functionals are local averages. Reconstruction of functions
from local averages was, for example, studied in \cite{sun02,sun02a,sun03,song06}.

\subsection{Approximation for General Measurement Functionals and Oversampling}\label{46}

The next theorem describes the convergence behavior of the approximation process
\eqref{45} in the case of oversampling.
\begin{theorem}\label{47}
  Let $0<\sigma<\pi$. There exists a complete orthonormal system $\{\Ft{\theta}_n\}_{n
    \in \N}$ in $L^2[-\pi,\pi]$ satisfying \eqref{44}, an
  associated sequence of measurement functionals $\{c_n\}_{n \in \N}$ as defined by
  \eqref{43}, and a constant
  $C_{6}$ such that for all stable LTI systems
  $T \colon \PWs_{\pi}^{1} \to \PWs_{\pi}^{1}$ and all $f \in \PWs_{\sigma}^{1}$ we have
  \begin{equation*}
    \sup_{t \in \R} \left| \sum_{n=1}^{N} c_n(f) (T \theta_n)(t) \right|
    \leq
    C_{6} \lVert f \rVert_{\PWs_{\sigma}^{1}}
    \lVert T \rVert
  \end{equation*}
  for all $N \in \N$, and further
  \begin{equation}\label{48}
    \lim_{N \rightarrow \infty} \left( \sup_{t \in \R} \left| (Tf)(t) - \sum_{n=1}^{N}
        c_n(f) (T \theta_n)(t) \right| \right)
    =
    0 .
  \end{equation}
\end{theorem}

\begin{remark}
  Theorem~\ref{47} shows that, using oversampling and
  more general measurement functionals, it is possible to have a stable system
  approximation with the process \eqref{45}. This
  is in contrast to pointwise sampling, which was analyzed in
  Section~\ref{19}, where even oversampling is not able to prevent
  the divergence. It is interesting to note that
  Theorem~\ref{47} is not only an abstract existence
  result. The complete orthonormal system $\{\Ft{\theta}_n\}_{n \in \N}$ which is used in
  Theorem~\ref{47} can be explicitly constructed by a
  procedure given in \cite{olevskii66a,olevskii66}.
\end{remark}

\begin{remark}
  In Section~\ref{57} we will see that oversampling is necessary in order to
  obtain Theorem~\ref{47}, i.e., a stable system
  implementation is only possible with oversampling and suitable measurement functionals.
\end{remark}

\begin{remark}\label{49}
  Theorem~\ref{47} also shows that, for the space
  $\PWs_{\pi}^{1}$, it is sufficient to use a linear process for the system approximation
  if oversampling is used, which introduces a kind of redundance. However, for other
  Banach spaces this is not necessarily true. There exist Banach spaces where non-linear
  processes have to be used, even in the signal reconstruction problem
  \cite{partington96}.
\end{remark}

For the proof of Theorem~\ref{47} we need the following theorem from
\cite{olevskii66a,olevskii66}.
\begin{theorem}[Olevskii]\label{50}
  Let $0<\delta<1$. There exists an orthonormal system $\{\psi_n\}_{n \in \N}$ of
  real-valued functions that is
  closed in $C[0,1]$ such that
  \begin{equation*}
    \sup_{n \in \N} \lVert \psi_n \rVert_{L^\infty[0,1]}
    <
    \infty
  \end{equation*}
  and such that there exists a constant $C_{7}$ such that for all $x \in
  [\delta,1]$ and all $N \in \N$ we have
  \begin{equation*}
    \int_{0}^{1} \left| \sum_{n=1}^{N} \psi_n(x) \psi_n(\tau) \right| \di{\tau}
    \leq
    C_{7} .
  \end{equation*}
\end{theorem}
\begin{remark}\label{51}
  In the above theorem, we adopted the notion of ``closed'' from \cite{olevskii75_book}.
  In \cite{olevskii75_book} a system $\{\psi_n\}_{n \in \N}$ is called closed in $C[0,1]$
  if every function in $C[0,1]$ can be uniformly approximated by finite linear combinations of
  the system $\{\psi_n\}_{n \in \N}$, that is if
  for every $\epsilon>0$ and every $f \in C[0,1]$ there exists an $N \in \N$ and a
  sequence $\{\alpha_n\}_{n=1}^{N} \subset \C$
  such that
  $\left\| f - \sum_{n=1}^{N} \alpha_n \psi_n  \right\|_{L^\infty[0,1]}
    <
    \epsilon$.
\end{remark}

\begin{proof}[Theorem~\ref{47}]
  \smartqed
  Let $0<\sigma<\pi$ be arbitrary but fixed and set
  $\delta=(\pi-\sigma)/(2\pi)$.
  Using the functions $\psi_n$ from Theorem~\ref{50}, we define
  \begin{equation*}
    \Ft{\theta}_n(\omega)
    \defequal
    \psi_n\left( \frac{\omega+\pi}{2\pi} \right) , \qquad \omega \in[-\pi,\pi] .
  \end{equation*}
  Due to the properties of the functions $\psi_n$, we see that
  $\{\Ft{\theta}_n\}_{n \in \N}$ is a complete orthonormal system for $L^2[-\pi,\pi]$, and that
  \begin{equation*}
    \sup_{n \in \N} \lVert \Ft{\theta}_n \rVert_{L^\infty[-\pi,\pi]}
    <
    \infty .
  \end{equation*}
  Furthermore, for $\omega \in [-\sigma,\sigma]$, we have
  \begin{align}
    \frac{1}{2 \pi} \int_{-\pi}^{\pi} \left| \sum_{n=1}^{N} \Ft{\theta}_n(\omega)
      \Ft{\theta}_n(\omega_1) \right| \di{\omega_1}
    &=
    \int_{0}^{1} \left| \sum_{n=1}^{N} \Ft{\psi}_n \left( \frac{\omega+\pi}{2 \pi} \right)
      \psi_n(\tau) \right| \di{\tau} \notag \\
    &\leq
    C_{7} , \label{52}
  \end{align}
  according to Theorem~\ref{50}, because for 
  $\omega \in [-\sigma,\sigma]$ we have $(\omega+\pi)/(2\pi) \in [\delta,1]$.
  Next, we study for $f \in \PWs_{\sigma}^{1}$ the expression
  \begin{align*}
    (U_N \Ft{f})(\omega)
    &\defequal
    \sum_{n=1}^{N} c_n(f) \Ft{\theta}_n(\omega) \\
    &=
    \frac{1}{2\pi} \int_{-\sigma}^{\sigma} \Ft{f}(\omega_1) \sum_{n=1}^{N}
    \Ft{\theta}_n(\omega) \Ft{\theta}_n(\omega_1) \di{\omega_1} .
  \end{align*}
  We have
  \begin{equation*}
    \lvert (U_N \Ft{f})(\omega) \rvert
    \leq
    \frac{1}{2\pi} \int_{-\sigma}^{\sigma} \lvert \Ft{f}(\omega_1) \rvert
    \left| \sum_{n=1}^{N}
    \Ft{\theta}_n(\omega) \Ft{\theta}_n(\omega_1) \right| \di{\omega_1} ,
  \end{equation*}
  which implies, using Fubini's theorem and \eqref{52}, that
  \begin{align}
    \frac{1}{2\pi} \int_{-\pi}^{\pi} \lvert (U_N \Ft{f})(\omega) \rvert \di{\omega}
    &\leq
    \frac{1}{2\pi} \int_{-\sigma}^{\sigma} \lvert \Ft{f}(\omega_1) \rvert
    \left( \frac{1}{2\pi} \int_{-\pi}^{\pi} \left| \sum_{n=1}^{N}
        \Ft{\theta}_n(\omega) \Ft{\theta}_n(\omega_1) \right| \di{\omega} \right) \di{\omega_1} \notag \\
    &\leq
    C_{7} \lVert f \rVert_{\PWs_{\sigma}^{1}} . \label{53}
  \end{align}

  Now, let $f \in \PWs_{\sigma}^{1}$ and $\epsilon>0$ be arbitrary but fixed.
  Then there exists an $f_\epsilon \in \PWs_{\sigma}^{2}$ such that
  \begin{equation}\label{54}
    \lVert f-f_\epsilon \rVert_{\PWs_{\sigma}^{1}}
    <
    \epsilon .
  \end{equation}
  We have
  \begin{align*}
    \frac{1}{2\pi} \int_{-\pi}^{\pi} & \lvert \Ft{f}(\omega) - (U_N \Ft{f})(\omega) \rvert
    \di{\omega} \\
    &\leq
    \frac{1}{2\pi} \int_{-\pi}^{\pi} \lvert \Ft{f}(\omega) - \Ft{f}_\epsilon(\omega)
    \rvert \di{\omega}
    +
    \frac{1}{2\pi} \int_{-\pi}^{\pi} \lvert \Ft{f}_\epsilon(\omega) - (U_N \Ft{f}_\epsilon)(\omega)
    \rvert \di{\omega} \\
    &\qquad+
    \frac{1}{2\pi} \int_{-\pi}^{\pi} \lvert (U_N (\Ft{f}-\Ft{f}_\epsilon))(\omega)
    \rvert \di{\omega} \\
    &\leq
    \epsilon + C_{7} \epsilon
    +
    \left( \frac{1}{2\pi} \int_{-\pi}^{\pi} \lvert \Ft{f}_\epsilon(\omega) - (U_N \Ft{f}_\epsilon)(\omega)
    \rvert^2 \di{\omega} \right)^{\frac{1}{2}} ,
  \end{align*}
  where we used \eqref{53} and \eqref{54}.
  Since $\PWs_{\sigma}^{2} \subset \PWs_{\pi}^{2}$ and $\{\Ft{\theta}_n\}_{n \in \N}$ is a
  complete orthonormal system in $L^2[-\pi,\pi]$, there exists a natural number
  $N_0=N_0(\epsilon)$ such that
  \begin{equation*}
    \left( \frac{1}{2\pi} \int_{-\pi}^{\pi} \lvert \Ft{f}_\epsilon(\omega) - (U_N \Ft{f}_\epsilon)(\omega)
      \rvert^2 \di{\omega} \right)^{\frac{1}{2}}
    <
    \epsilon
  \end{equation*}
  for all $N \geq N_0$. Hence, we have
  \begin{equation*}
    \frac{1}{2\pi} \int_{-\pi}^{\pi} \lvert \Ft{f}(\omega) - (U_N \Ft{f})(\omega) \rvert
    \di{\omega}
    \leq
    \epsilon(2+ C_{7})
  \end{equation*}
  for all $N \geq N_0$.
  This shows that
  \begin{equation}\label{55}
    \lim_{N \rightarrow \infty}
    \frac{1}{2\pi} \int_{-\pi}^{\pi} \lvert \Ft{f}(\omega) - (U_N \Ft{f})(\omega) \rvert
    \di{\omega}
    =
    0 .
  \end{equation}

  Next, let $T \colon \PWs_{\pi}^{1} \to \PWs_{\pi}^{1}$ be an arbitrary but fixed stable LTI
  system.
  We have
  \begin{align*}
    (Tf)(t)
    &-
    \sum_{n=1}^{N} c_n(f) (T\theta_n)(t) \\
    &=
    \frac{1}{2\pi}
    \int_{-\pi}^{\pi} \left(
      \Ft{f}(\omega) \Ft{h}_T(\omega) \e^{\iu \omega t}
      -
      \sum_{n=1}^{N} c_n(f) \Ft{h}_T(\omega) \Ft{\theta}_n(\omega) \e^{\iu \omega t}
    \right) \di{\omega} \\
    &=
    \frac{1}{2\pi}
    \int_{-\pi}^{\pi} (
      \Ft{f}(\omega) 
      -
      (U_N \Ft{f})(\omega)
    ) \Ft{h}_T(\omega) \e^{\iu \omega t} \di{\omega}
  \end{align*}
  and consequently
  \begin{equation}\label{56}
    \left|
      (Tf)(t)
      -
      \sum_{n=1}^{N} c_n(f) (T\theta_n)(t)
    \right|
    \leq
    \lVert \Ft{h}_T \rVert_{L^\infty[-\pi,\pi]}
    \frac{1}{2\pi}
    \int_{-\pi}^{\pi}
    \lvert 
    \Ft{f}(\omega) 
    -
    (U_N \Ft{f})(\omega)
    \rvert \di{\omega} 
  \end{equation}
  for all $t \in \R$.
  From \eqref{55} and \eqref{56} we see
  that
  \begin{equation*}
    \lim_{N \rightarrow \infty} \left( \sup_{t \in \R}
    \left|
      (Tf)(t)
      -
      \sum_{n=1}^{N} c_n(f) (T\theta_n)(t)
    \right| \right)
    =
    0 .
  \end{equation*}
  
  Further, we have
  \begin{align*}
    \left|
      \sum_{n=1}^{N} c_n(f) (T\theta_n)(t)
    \right|
    &\leq
    \frac{1}{2\pi}
    \int_{-\pi}^{\pi}
    \lvert (U_N \Ft{f})(\omega)  
    \Ft{h}_T(\omega) \rvert \di{\omega} \\
    &\leq
    C_{7} \lVert \Ft{h}_T \rVert_{L^\infty[-\pi,\pi]} \lVert f \rVert_{\PWs_{\sigma}^{1}} , 
  \end{align*}
  where we used \eqref{53} in the last inequality. \qed
\end{proof}

\begin{remark}
  Since $\{\Ft{\theta}_n\}_{n \in \N}$ is a complete orthonormal system in
  $L^2[-\pi,\pi]$, it follows that the projections of the functions $\{\theta_n\}_{n \in
    \N}$ onto $\PWs_{\sigma}^{2}$ form a Parseval frame for $\PWs_{\sigma}^{2}$,
  $0<\sigma<\pi$ \cite[p.~231]{heil11_book}. Although we have seen in
  Remark~\ref{22} that a frame does not necessary lead to a convergent
  approximation process, Theorem~\ref{47}
  shows that there are even Parseval frames for which we have convergence.
\end{remark}

\subsection{The Necessity of Oversampling}\label{57}

In Section~\ref{46} we have seen that if oversampling
and generalized measurement functionals are used, we can approximate $Tf$ by
\eqref{45}. The question whether this remains true
if no oversampling is used, is the subject of this section. We want to answer this
question for a large class of practically relevant measurement functionals.

We start with a biorthogonal system $\{\Ft{\gamma}_n,\Ft{\phi}_n\}_{n \in \N}$, i.e., a system
that satisfies
\begin{equation*}
  \frac{1}{2\pi} \int_{-\pi}^{\pi} \Ft{\gamma}_n(\omega) \cc{\Ft{\phi}_m(\omega)} \di{\omega}
  =
  \begin{cases}
    1,& m=n,\\
    0,& m \neq n .
  \end{cases}
\end{equation*}
Further, we assume that $\{\Ft{\gamma}_n\}_{n \in \N} \subset L^\infty[-\pi,\pi]$ and
$\{\phi_n \}_{n \in \N} \subset \PWs_{\pi}^{2}$, and define the measurement functionals by
\begin{equation}\label{58}
  c_n(f)
  \defequal
  \frac{1}{2\pi} \int_{-\pi}^{\pi} \Ft{f}(\omega) \cc{\Ft{\gamma}_n(\omega)} \di{\omega} ,
  \quad n\in \N .
\end{equation}
As discussed in Section~\ref{40}, we additionally require that
\begin{equation}\label{59}
  \sup_{n \in \N} \lVert \Ft{\gamma}_n \rVert_{L^\infty[-\pi,\pi]}
  <
  \infty ,
\end{equation}
in order that \eqref{58} defines reasonable measurement
functionals for $f \in \PWs_{\pi}^{1}$.
We further assume that there exists a constant $C_{8}$ such that
for any finite sequence $\{a_n\}$ we have
\begin{equation}\label{60}
  \int_{-\infty}^{\infty} \left| \sum_{n} a_n \phi_n(t) \right|^2 \di{t}
  \leq
  C_{8} \sum_{n} \lvert a_n \rvert^2 .
\end{equation}
Condition \eqref{60} relates the $l^2$-norm of the coefficients to
the $L^2(\R)$-norm of the continuous-time signal. If \eqref{60} is
fulfilled, the $L^2(\R)$-norm of the continuous-time signal is always bounded above by the
$l^2$-norm of the coefficients, i.e., the measurement values. This property is practically
interesting, because in digital signal processing we operate on the sequence of
coefficients $\{a_n\}_{n \in \N}$ by using stable $l^2 \to l^2$ mappings, and we always
want to be able to control the $L^2(\R)$-norm of the corresponding continuous-time signal.
Note that in the special case of equidistant pointwise sampling at Nyquist rate, the norms
are equal according to Parseval's equality.

\begin{remark}
  Instead of requiring \eqref{60} to hold we could also require that
  there exists a constant $C_{9}$ such that for any finite sequence $\{a_n\}$
  we have
  \begin{equation}\label{61}
    \int_{-\infty}^{\infty} \left| \sum_{n} a_n \gamma_n(t) \right|^2 \di{t}
    \geq
    C_{9} \sum_{n} \lvert a_n \rvert^2 .
  \end{equation}
  Indeed \eqref{61} is a weaker assumption than \eqref{60},
  because condition \eqref{60} implies condition \eqref{61}
  but the reverse direction is not true in general.
\end{remark}

\begin{remark}
  Note that the setting which we consider here is a generalization of the setting that
  arises when staring with complete interpolating sequences.
\end{remark}

\begin{theorem}\label{62}
  Let $\{\Ft{\gamma}_n,\Ft{\phi}_n\}_{n \in \N}$ be a biorthogonal system
  that satisfies \eqref{59} and
  \eqref{60}, and let $\{c_n\}_{n \in \N}$ be the associated sequence of
  measurement functionals as defined by \eqref{58}. For
  every $t \in \R$ there exist a stable LTI system $T_* \colon \PWs_{\pi}^{1} \to
  \PWs_{\pi}^{1}$ and a signal $f_* \in \PWs_{\pi}^{1}$ such that
  \begin{equation}\label{63}
    \limsup_{N \rightarrow \infty} \left| \sum_{n=1}^{N}
      c_n(f_*)
      (T_* \phi_n)(t)
    \right|
    =
    \infty .
  \end{equation}
\end{theorem}

\begin{remark}
  The orthonormal sequence from Section~\ref{40} of course satisfies the
  conditions of Theorem~\ref{62}. This shows how important the
  assumption of oversampling, i.e., $f \in \PWs_{\sigma}^{1}$, $\sigma < \pi$, is in order
  to obtain Theorem~\ref{47}. 
\end{remark}

For the proof we use the following result from \cite{szarek80}, which is included here for
convenience, with a slightly modified notation.
\begin{proposition}[Szarek]\label{64}
  Let $(S,\mathcal{B},m)$ a probability space and $\{f_n,g_n\}_{n \in \N}$ a biorthogonal
  sequence of measurable functions on $S$ (i.e., $\int_S f_k \cc{g_n} \di{m} = \delta_{kn}$)
  such that
  \begin{enumerate}
  \item $\lVert g_n \rVert_{\infty} \leq 1$ for $n=1,2,\ldots,N$.
  \item $\int_S \lvert \sum_{n=1}^{N} s_n f_n \rvert^2 \di{m} \leq C \sum_{n=1}^{N} \lvert
    s_n \rvert^2$ for some $C>0$ and for all sequences of scalars $s_1, \ldots,s_N$ (and,
    as a consequence, $\int_S \lvert \sum_{n=1}^{N} t_n g_n \rvert^2 \di{m} \geq C^{-1}
    \sum_{n=1}^{N} \lvert t_n \rvert^2$ for all scalars $t_1, \ldots , t_N)$.
  \end{enumerate}
  Then there exists $C'>0$, depending only on $C$, such that
  \begin{equation*}
    \max_{1 \leq M \leq N} \int_S \int_S \left| \sum_{n=1}^{M} \cc{g_n(t)} f_n(s) \right| 
    \di{m}(t) \di{m}(s)
    \geq
    C' \log(N) .
  \end{equation*}
\end{proposition}

\begin{proof}[Theorem~\ref{62}]
  \smartqed
  Let $\{\Ft{\gamma}_n,\Ft{\phi}_n\}_{n \in \N}$ be an arbitrary but fixed biorthogonal system
  that satisfies \eqref{59} and
  \eqref{60}.
  According to Proposition~\ref{64} we have
  \begin{equation*}
    \max_{1 \leq M \leq N}
    \left( \frac{1}{2\pi} \right)^2
    \int_{-\pi}^{\pi} \int_{-\pi}^{\pi} \left| \sum_{n=1}^{M} \cc{\Ft{\gamma}_n(\omega)}
      \Ft{\phi}_n(\omega_1) \right| \di{\omega} \di{\omega_1}
    \geq
    C_{10} \log(N)
  \end{equation*}
  with a universal constant $C_{10}$. This implies that
  \begin{equation*}
    \max_{1 \leq M \leq N}
    \esssup_{\omega \in [-\pi,\pi]}
    \frac{1}{2\pi}
    \int_{-\pi}^{\pi} \left| \sum_{n=1}^{M} \cc{\Ft{\gamma}_n(\omega)}
      \Ft{\phi}_n(\omega_1) \right| \di{\omega_1}
    \geq
    C_{10} \log(N) .
  \end{equation*}
  As in the proof of Theorem~\ref{20} it is shown that
  there exists a stable LTI system $T_* \colon \PWs_{\pi}^{1} \to \PWs_{\pi}^{1}$
  such that
  \begin{equation*}
    \limsup_{N \rightarrow \infty}
    \esssup_{\omega \in [-\pi,\pi]}
    \left| \sum_{n=1}^{N} \cc{\Ft{\gamma}_n(\omega)} (T_* \phi_n)(t) \right|
    =
    \infty .
  \end{equation*}
  And again by the same reasoning as in the proof of
  Theorem~\ref{20}, there exists a signal $f_* \in
  \PWs_{\pi}^{1}$ such that
  \begin{equation*}
    \limsup_{N \rightarrow \infty}
    \left| \sum_{n=1}^{N} \left( \frac{1}{2\pi} \int_{-\pi}^{\pi}  \Ft{f}_*(\omega)
      \cc{\Ft{\gamma}_n(\omega)} \di{\omega} \right) (T_* \phi_n)(t) \right|
    =
    \infty .
  \end{equation*}
  This completes the proof.  
  \qed
\end{proof}

\section{Convergence of Subsequences of Certain Measurement Procedures}\label{65}

So far, we have seen that a system approximation is possible if we use suitable
measurement functionals and oversampling. Further, the previous section has shown that
oversampling is necessary, because without oversampling we can always find a stable LTI
system $T_* \colon \PWs_{\pi}^{1} \to \PWs_{\pi}^{1}$ and a signal $f_* \in
\PWs_{\pi}^{1}$ such that \eqref{63} is true. Since in
\eqref{63} we have a $\limsup$, it is legitimate to ask
whether there exists an increasing subsequence $\{M_N\}_{N \in \N}$ of the natural numbers
such that
\begin{equation}\label{66}
  \lim_{N \rightarrow \infty} \left| (Tf)(t) - \sum_{n=1}^{M_N} c_n(f) (T \phi_n)(t)
  \right|
  =
  0 .
\end{equation}
If \eqref{66} was true it would show that a careful choice of the number of measurements
that are used in each step of the approximation could generate a convergent approximation
process, even without oversampling. Theorem~\ref{71} will answer this question
in the affirmative for a special pair of measurement functionals and reconstruction
functions.

For $k \in \N_0=\N \cup \{0\}$ we consider the functions
\begin{equation}\label{67}
  \Ft{\theta}_k(\omega)
  =
  w_k \left( \frac{\omega+\pi}{2\pi} \right) , \quad -\pi \leq \omega < \pi ,
\end{equation}
where $w_k$ are the Walsh functions.
Then $\{\Ft{\theta}_k\}_{k \in \N_0}$ is a complete orthonormal system in $L^2[-\pi,\pi]$.
Further, let $T\colon \PWs_{\pi}^{1} \to \PWs_{\pi}^{1}$ be a stable LTI system.
For $t \in \R$ we define
\begin{equation}\label{68}
  c_k(f,t)
  \defequal
  \frac{1}{2\pi} \int_{-\pi}^{\pi} \Ft{f}(\omega) \Ft{\theta}_k(\omega) \e^{\iu \omega t}
  \di{\omega} ,
\end{equation}
and analyze the convergence behavior of
\begin{equation}\label{69}
  \sum_{k=0}^{2^N} c_k(f,0) (T \theta_k)(t)
\end{equation}
and 
\begin{equation}\label{70}
  \sum_{k=0}^{2^N} c_k(f,t) (T \theta_k)(0) 
\end{equation}
as $N$ tends to infinity. In \eqref{69} we have the
ordinary system approximation process, except for the difference that the number of
measurements, and consequently the number of summands used for the approximation, is
doubled in each approximation step. In \eqref{70} we have
an alternative implementation of the system, where the variable $t$ is included in the
measurement functionals. As in \eqref{69}, the number of
measurements is doubled in each step.

We have the following result.
\begin{theorem}\label{71}
  Let $\{\theta_k\}_{k \in \N_0}$ be defined through its Fourier transform
  \eqref{67} and $c_k$ as in \eqref{68}. For all $f \in \PWs_{\pi}^{1}$
  and all stable LTI systems $T \colon \PWs_{\pi}^{1} \to \PWs_{\pi}^{1}$ we have
  \begin{equation}\label{72}
    \lim_{N \rightarrow \infty} \left( \sup_{t \in \R} \left| (Tf)(t) - \sum_{k=0}^{2^N}
        c_k(f,0) (T \theta_k)(t) \right| \right)
    =
    0
  \end{equation}
  and
  \begin{equation}\label{73}
    \lim_{N \rightarrow \infty} \left( \sup_{t \in \R} \left| (Tf)(t) - \sum_{k=0}^{2^N}
        c_k(f,t) (T \theta_k)(0) \right| \right)
    =
    0 .
  \end{equation}
\end{theorem}

Theorem~\ref{71} shows that there exists a complete orthonormal system that
leads to a stable system approximation process for all $f \in \PWs_{\pi}^{1}$ and all
stable LTI systems $T\colon \PWs_{\pi}^{1} \to \PWs_{\pi}^{1}$ if we restrict to a
suitable subsequence. It is important to note that the subsequence is universal because it
neither depends on the signal $f$ nor on the system $T$. It is also interesting that with
this kind of approximation we do not need oversampling in order to have convergence.

\begin{remark}\label{74}
  For sampling-based signal processing with equidistant sampling points at Nyquist rate
  such a result cannot exist, because for every subsequence $\{M_N\}_{N \in N}$ of the
  natural numbers there exists a signal
  $f_* \in \PWs_{\pi}^{1}$ and stable LTI system $T_* \colon \PWs_{\pi}^{1} \to
  \PWs_{\pi}^{1}$ such that
  \begin{equation*}
    \limsup_{N \rightarrow \infty} \left( \sup_{t \in \R} \left| (T_*f_*)(t) - \sum_{k=-M_N}^{M_N}
        f_*(k) (T_* \sinc(\spacedot -k))(t) \right| \right)
    =
    \infty .
  \end{equation*}
  This follows directly from the fact that there exists a positive
  constant $C_{11}$ such that
  \begin{equation}\label{75}
    \frac{1}{2 \pi} \int_{-\pi}^{\pi} \left| \sum_{k=-N}^{N} \e^{\iu k(\omega - \omega_1)}
    \right| \di{\omega_1}
    \geq
    C_{11} \log(N)
  \end{equation}
  for all $\omega \in [-\pi,\pi]$ and all $N \in \N$ \cite[p.~67]{zygmund02_vol1_book}.
\end{remark}

\begin{proof}[Theorem~\ref{71}]
  \smartqed Let $\{\theta_k\}_{k \in \N_0}$ be defined through its Fourier transform
  \eqref{67} and $c_k$ as in \eqref{68}. Further, let $f \in
  \PWs_{\pi}^{1}$ and $T \colon \PWs_{\pi}^{1} \to \PWs_{\pi}^{1}$ be a stable LTI system,
  both arbitrary but fixed.

  We first prove \eqref{72}.  
  In \cite{fine49} it was shown that
  \begin{equation*}
    \lim_{N \rightarrow \infty} \frac{1}{2\pi} \int_{-\pi}^{\pi} \left| \Ft{f}(\omega) -
      \sum_{k=0}^{2^N} c_k(f,0) \Ft{\theta}_k(\omega) \right| \di{\omega}
    =
    0 .
  \end{equation*}
  Further, since
  \begin{align*}
    &\left|
    (T f)(t)
    -
    \sum_{k=0}^{2^N} c_k(f,0) (T \theta_k)(t)
    \right|\\
    &=
    \left|
    \frac{1}{2\pi} \int_{-\pi}^{\pi} \Ft{f}(\omega) \Ft{h}_T(\omega) \e^{\iu \omega t}
    -
    \sum_{k=0}^{2^N} c_k(f,0) \Ft{\theta}_k(\omega) \Ft{h}_T(\omega) \e^{\iu \omega t}
    \di{\omega}
    \right| \\
    &=
    \left|
    \frac{1}{2\pi} \int_{-\pi}^{\pi}
    \left(
      \Ft{f}(\omega) - \sum_{k=0}^{2^N} c_k(f,0) \Ft{\theta}_k(\omega)
    \right)
    \Ft{h}_T(\omega) \e^{\iu \omega t}
    \di{\omega}
    \right| \\
    &\leq
    \lVert \Ft{h}_T \rVert_{L^\infty[-\pi,\pi]}
    \frac{1}{2\pi} \int_{-\pi}^{\pi}
    \left|
      \Ft{f}(\omega) - \sum_{k=0}^{2^N} c_k(f,0) \Ft{\theta}_k(\omega)
    \right|
    \di{\omega} ,
  \end{align*}
  the first assertion of the theorem is proved.
 
  Next, we prove \eqref{73}.
  Let $\epsilon>0$ be arbitrary but fixed. There exists a measurable set $F_\epsilon
  \subset [-\pi,\pi]$ such that
  \begin{equation*}
    \frac{1}{2\pi} \int_{F_\epsilon} \lvert \Ft{f}(\omega) \rvert
    \di{\omega}
    <
    \frac{\epsilon}{2}
  \end{equation*}
  and
  \begin{equation*}
    \esssup_{\omega \in [-\pi,\pi] \setminus F_\epsilon} \lvert \Ft{f}(\omega) \rvert
    =
    C(\Ft{f},F_\epsilon)
    <
    \infty .
  \end{equation*}
  Further, we have
  \begin{align}
    &\left|
      \frac{1}{2\pi} \int_{-\pi}^{\pi}
      \left(
        \Ft{h}_T(\omega) - \sum_{k=0}^{2^N} \Ft{\theta}_k(\omega) (T\theta_k)(0)
      \right)
      \Ft{f}(\omega) \e^{\iu \omega t}
      \di{\omega}
    \right| \notag \\
    &\leq
    \frac{1}{2\pi} \int_{-\pi}^{\pi}
    \left|
      \Ft{h}_T(\omega) - \sum_{k=0}^{2^N} \Ft{\theta}_k(\omega) (T\theta_k)(0)
    \right|
    \lvert \Ft{f}(\omega) \rvert
    \di{\omega} \notag \\
    &=
    \frac{1}{2\pi} \int_{F_\epsilon}
    \left|
      \Ft{h}_T(\omega) - \sum_{k=0}^{2^N} \Ft{\theta}_k(\omega) (T\theta_k)(0)
    \right|
    \lvert \Ft{f}(\omega) \rvert
    \di{\omega} \notag \\
    &\quad+
    \frac{1}{2\pi} \int_{[-\pi,\pi] \setminus F_\epsilon}
    \left|
      \Ft{h}_T(\omega) - \sum_{k=0}^{2^N} \Ft{\theta}_k(\omega) (T\theta_k)(0)
    \right|
    \lvert \Ft{f}(\omega) \rvert
    \di{\omega} .\label{76}
  \end{align}
  Next, we analyze the two summands on the right hand side of \eqref{76}. 
  For the first summand we have
  \begin{align}
    &\frac{1}{2\pi} \int_{F_\epsilon}
    \left|
      \Ft{h}_T(\omega) - \sum_{k=0}^{2^N} \Ft{\theta}_k(\omega) (T\theta_k)(0)
    \right|
    \lvert \Ft{f}(\omega) \rvert
    \di{\omega} \notag \\
    &\leq
    \frac{1}{2\pi} \int_{F_\epsilon}
    \lvert \Ft{h}_T(\omega) \rvert
    \lvert \Ft{f}(\omega) \rvert
    \di{\omega}
    +
    \frac{1}{2\pi} \int_{F_\epsilon}
    \left|
      \sum_{k=0}^{2^N} \Ft{\theta}_k(\omega) (T\theta_k)(0)
    \right|
    \lvert \Ft{f}(\omega) \rvert
    \di{\omega} \notag \\
    &\leq
    2 \lVert \Ft{h}_T \rVert_{L^\infty[-\pi,\pi]} \frac{1}{2\pi} \int_{F_\epsilon} \lvert
    \Ft{f}(\omega) \rvert \di{\omega} \notag \\
    &<
    \epsilon \lVert \Ft{h}_T \rVert_{L^\infty[-\pi,\pi]} , \label{77}
  \end{align}
  because
  \begin{align*}
    \left|
      \sum_{k=0}^{2^N} \Ft{\theta}_k(\omega) (T\theta_k)(0)
    \right|
    &=
    \left|
      \frac{1}{2\pi} \int_{-\pi}^{\pi} \Ft{h}_T(\omega_1) \sum_{k=0}^{2^N}
      \Ft{\theta}_k(\omega) \Ft{\theta}_k(\omega_1) \di{\omega_1}
    \right| \\
    &\leq
    \lVert \Ft{h}_T \rVert_{L^\infty[-\pi,\pi]}
    \frac{1}{2\pi} \int_{-\pi}^{\pi} \left| \sum_{k=0}^{2^N}
      \Ft{\theta}_k(\omega) \Ft{\theta}_k(\omega_1) \right| \di{\omega_1}
  \end{align*}
  and
  \begin{equation}\label{78}
    \frac{1}{2\pi} \int_{-\pi}^{\pi} \left| \sum_{k=0}^{2^N} \Ft{\theta}_k(\omega)
    \Ft{\theta}_k(\omega_1) \right| \di{\omega_1}
    =
    1
  \end{equation}
  for all $\omega \in [-\pi,\pi]$ \cite{fine49,schipp90_book}.
  For the second summand we have
  \begin{align*}
    &\frac{1}{2\pi} \int_{[-\pi,\pi] \setminus F_\epsilon}
    \left|
      \Ft{h}_T(\omega) - \sum_{k=0}^{2^N} \Ft{\theta}_k(\omega) (T\theta_k)(0)
    \right|
    \lvert \Ft{f}(\omega) \rvert
    \di{\omega} \\
    &\leq
    C(\Ft{f},F_\epsilon)
    \frac{1}{2\pi} \int_{[-\pi,\pi] \setminus F_\epsilon}
    \left|
      \Ft{h}_T(\omega) - \sum_{k=0}^{2^N} \Ft{\theta}_k(\omega) (T\theta_k)(0)
    \right| \di{\omega} \\
    &\leq
    C(\Ft{f},F_\epsilon)
    \left( \frac{1}{2\pi} \int_{-\pi}^{\pi}
    \left|
      \Ft{h}_T(\omega) - \sum_{k=0}^{2^N} \Ft{\theta}_k(\omega) (T\theta_k)(0)
    \right|^2 \di{\omega} \right)^{\frac{1}{2}} ,
  \end{align*}
  and because
  \begin{equation*}
    \lim_{N \rightarrow \infty} \frac{1}{2\pi} \int_{-\pi}^{\pi}
    \left|
       \Ft{h}_T(\omega) - \sum_{k=0}^{2^N} \Ft{\theta}_k(\omega) (T\theta_k)(0)
    \right|^2 \di{\omega}
    =
    0 ,
  \end{equation*}
  according to \cite{schipp90_book},
  there exists a natural number
  $N_0=N_0(\epsilon)$ such that
  \begin{equation}\label{79}
    \frac{1}{2\pi} \int_{[-\pi,\pi] \setminus F_\epsilon}
    \left|
      \Ft{h}_T(\omega) - \sum_{k=0}^{2^N} \Ft{\theta}_k(\omega) (T\theta_k)(0)
    \right|
    \lvert \Ft{f}(\omega) \rvert
    \di{\omega}
    <
    \epsilon
  \end{equation}
  for all $N \geq N_0$.
  Combining \eqref{76}, \eqref{77}, and \eqref{79}, we see that
  \begin{equation*}
    \left|
      \frac{1}{2\pi} \int_{-\pi}^{\pi}
      \left(
        \Ft{h}_T(\omega) - \sum_{k=0}^{2^N} \Ft{\theta}_k(\omega) (T\theta_k)(0)
      \right)
      \Ft{f}(\omega) \e^{\iu \omega t}
      \di{\omega}
    \right|
    \leq
    (\lVert \Ft{h}_T \rVert_{L^\infty[-\pi,\pi]}+1) \epsilon
  \end{equation*}
  for all $N \geq N_0$, and
  since
  \begin{align*}
    &\left| (Tf)(t) - \sum_{k=0}^{2^N}
      c_k(f,t) (T \theta_k)(0) \right| \\
    &=
    \left|
      \frac{1}{2\pi} \int_{-\pi}^{\pi}
      \left(
        \Ft{h}_T(\omega) - \sum_{k=0}^{2^N} \Ft{\theta}_k(\omega) (T\theta_k)(0)
      \right)
      \Ft{f}(\omega) \e^{\iu \omega t}
      \di{\omega}
    \right| ,
  \end{align*}
  the proof is complete.
  \qed
\end{proof}

Theorem~\ref{71} shows that if the summation is restricted to a suitable
subsequence of the natural numbers, we can have a convergent system approximation process
if we use measurement functionals. Now the question arises if this is also true for
pointwise sampling as analyzed in Section~\ref{4}. Since in
Theorem~\ref{20} we only have a $\limsup$ this could be the
case. However, we have the following conjecture.
\begin{conjecture}\label{80}
  Let $\{t_k\}_{k \in \Z} \subset \R$ be an ordered complete interpolating sequence for $\PWs_{\pi}^{2}$.
  Then there exists a positive constant $C_{12}$ such that
  \begin{equation}\label{81}
    \frac{1}{2\pi} \int_{-\pi}^{\pi} \left| \sum_{k=-N}^{N} \e^{\iu \omega t_k}
      \Ft{\phi}_k(\omega_1) \right| \di{\omega_1}
    \geq
    C_{12} \log(N)
  \end{equation}
  for all $\omega \in [-\pi,\pi]$ and all $N \in \N$.
\end{conjecture}
If this conjecture is true then the derivations in this work imply that a theorem such as
Theorem~\ref{71} cannot hold for the sampling-based system approximation that
was treated in Section~\ref{4}. Because then, for every subsequence
$\{M_N\}_{N \in \N}$ of the natural numbers and all ordered complete interpolating
sequences $\{t_k\}_{k \in \Z} \subset \R$ we have
\begin{equation}\label{82}
  \limsup_{N \rightarrow \infty}
  \left( \sup_{t \in \R} \left| (T_* f_*)(t) - \sum_{k=-M_N}^{M_N}
      f_*(t_k) (T_* \phi_k)(t) \right| \right)
  =
  \infty
\end{equation}
for some $f_* \in \PWs_{\pi}^{1}$ and some stable LTI system $T_* \colon \PWs_{\pi}^{1}
\to \PWs_{\pi}^{1}$. In fact, in order to obtain this negative result for sampling-based
system approximation it would suffice to have an arbitrary sequence $\{L_N\}_{N \in \N}$
with $\lim_{N \rightarrow \infty} L_N = \infty$ on the right-hand side of
\eqref{81}. Note that we already know from \eqref{75}
and Remark~\ref{74} that Conjecture~\ref{80} and
\eqref{82} are true for the special case of equidistant sampling.

\section{More General Measurement Functionals}\label{83}

In this section we consider even more general measurement functionals than those in
Section~\ref{40}. For this, we restrict ourselves to stable LTI systems
$T$ with continuous $\Ft{h}_T$.

Now let $\{\Ft{g}_n\}_{n \in \N} \subset C[-\pi,\pi]$ be a sequence of functions with the
following properties:
\begin{enumerate}
\item\label{84} $\sup_{n \in \N} \lVert \Ft{g}_n \rVert_{L^\infty[-\pi,\pi]} < \infty$ and $\inf_{n \in \N}
\lVert \Ft{g}_n \rVert_{L^\infty[-\pi,\pi]} > 0$.
\item\label{85} $\{\Ft{g}_n\}_{n \in \N}$ is closed in $C[-\pi,\pi]$ and minimal, in
  the sense that for all $m \in \N$ the function $\Ft{g}_m$ is not in the closed span of
  $\{\Ft{g}_n\}_{n \neq m}$.
\item\label{86} There exists a constant
$C_{13}>0$ such that for any finite sequences
$\{a_n\}$ we have
\begin{equation}\label{87}
  \left\| \sum_{n} a_n \Ft{g}_n \right\|_{L^\infty[-\pi,\pi]}
  \geq
  \frac{1}{ C_{13}}
  \left( \sum_{n} \lvert a_n \rvert^2 \right)^{\frac{1}{2}} .
\end{equation}
\end{enumerate}

Property~\ref{85} guarantees that there exists a unique sequence of functionals
$\{u_n\}_{n \in \N}$ which is biorthogonal to $\{\Ft{g}_n\}_{n \in \N}$ \cite[p.~155]{heil11_book}.

We shortly discuss the structure of measurement functionals and approximation processes
which are based on sequences $\{\Ft{g}_n\}_{n \in \N} \subset C[-\pi,\pi]$ that satisfy
the properties \ref{84}--\ref{86}. Let $\{u_n\}_{n \in \N}$ be the unique
sequence of functionals which is biorthogonal to $\{\Ft{g}_n\}_{n \in \N}$. Since we
assume that $\Ft{h}_T \in C[-\pi,\pi]$, it follows that there exist finite regular Borel
measures $\mu_n$ such that
\begin{equation*}
  u_n(\Ft{h}_T)
  =
  \frac{1}{2\pi} \int_{-\pi}^{\pi} \Ft{h}_T(\omega) \di{\mu_n}(\omega) .
\end{equation*}
In \cite{szarek80} it was shown that, due to property \ref{86}, there exists a
regular Borel measure $\nu$ such that
\begin{equation*}
  \sum_{n=1}^{\infty} \lvert c_n(\Ft{h}_T) \rvert^2
  \leq
  C_{14} \int_{-\pi}^{\pi} \lvert \Ft{h}_T(\omega) \rvert^2 \di{\nu}(\omega) .
\end{equation*}
Further, all Borel measures $\mu_n$ are absolutely continuous with respect to $\nu$, and the
Radon--Nikodym derivatives of $\mu_n$ with respect to $\nu$, which we call $F_n$, are in
$L^2(\nu)$, i.e, we have
\begin{equation*}
  \int_{-\pi}^{\pi} \lvert F_n(\omega) \rvert^2 \di{\nu}(\omega)
  <
  \infty .
\end{equation*}
It follows that
\begin{equation*}
  \frac{1}{2\pi} \int_{-\pi}^{\pi} \Ft{g}_n(\omega) F_l(\omega) \di{\nu}(\omega)
  =
  \begin{cases}
    1,& n=l,\\
    0,& n \neq l,
  \end{cases}
\end{equation*}
i.e., the system $\{\Ft{g}_n,\cc{F_n}\}_{n \in \N}$ is a biorthogonal system with respect
to the measure $\nu$.

Note that this time we have a system that is
biorthogonal with respect to the regular Borel measure $\nu$ and not with respect to
the Lebesgue measure, as before. Thus, if we only require property
\ref{86}, we cannot find a corresponding biorthogonal system for the Lebesgue
measure in general, but only for more general measures. Nevertheless, we can obtain the
divergence result that is stated in Theorem~\ref{89}.

In \cite{szarek80} it was analyzed whether a basis for $C[-\pi,\pi]$ that satisfies the
above properties \ref{84}--\ref{86} could exist, and the nonexistence of
such a basis was proved. We employ this result to prove the following theorem, in
which we use the abbreviations
\begin{equation*}
  c_n(f,t)
  \defequal
  \frac{1}{2\pi} \int_{-\pi}^{\pi} \Ft{f}(\omega) \Ft{g}_n(\omega) \e^{\iu \omega t}
  \di{\omega} .
\end{equation*}
and
\begin{equation}\label{88}
  w_n(\Ft{h}_T,t)
  =
  \frac{1}{2 \pi} \int_{-\pi}^{\pi}  \Ft{h}_T(\omega) \e^{\iu \omega t} F_n(\omega)
  \di{\nu}(\omega) .
\end{equation}

\begin{theorem}\label{89}
  Let $\{\Ft{g}_n\}_{n \in \N} \subset C[-\pi,\pi]$ be an arbitrary sequence of functions
  that satisfies the above properties \ref{84}--\ref{86}, and let $t \in
  \R$. Then we have:
  \begin{enumerate}
  \item\label{90} There exists a stable LTI system $T_{*1} \colon \PWs_{\pi}^{1} \to \PWs_{\pi}^{1}$
  with $\Ft{h}_{T_{*1}} \in C[-\pi,\pi]$ and a signal $f_{*1} \in \PWs_{\pi}^{1}$ such that
  \begin{equation}\label{91}
    \limsup_{N \rightarrow \infty}
    \left| \sum_{n=1}^{N} c_n(f_{*1},t) w_n(\Ft{h}_{T_{*1}},0) \right|
    =
    \infty .
  \end{equation}

  \item\label{92} There exists a stable LTI system $T_{*2} \colon \PWs_{\pi}^{1} \to \PWs_{\pi}^{1}$
  with $\Ft{h}_{T_{*2}} \in C[-\pi,\pi]$ and a signal $f_{*2} \in \PWs_{\pi}^{1}$ such that
  \begin{equation}\label{93}
    \limsup_{N \rightarrow \infty}
    \left| \sum_{n=1}^{N} c_n(f_{*2},0) w_n(\Ft{h}_{T_{*2}},t) \right|
    =
    \infty .
  \end{equation}
  \end{enumerate}
\end{theorem}

\begin{proof}
  \smartqed
  We start with the proof of assertion \ref{90}.
  In \cite{szarek80} it was proved that there exists no basis for $C[-\pi,\pi]$ with the
  above properties \ref{84}--\ref{86}. That is, if we set
  \begin{equation*}
    (S_N \Ft{h}_T)(\omega)
    =
    \sum_{n=1}^{N} w_n(\Ft{h}_T,0) \Ft{g}_n(\omega), \quad \omega \in [-\pi,\pi],
  \end{equation*}
  then, for
  \begin{equation*}
    \lVert S_N \rVert
    =
    \sup_{
      \substack{\Ft{h}_T \in C[-\pi,\pi],\\
      \lVert \Ft{h}_T \rVert_{L^\infty[-\pi,\pi]} \leq 1}
    }
    \lVert S_N \Ft{h}_T \rVert_{L^\infty[-\pi,\pi]}
  \end{equation*}
  we have according to \cite{szarek80} that
  \begin{equation*}
    \limsup_{N \rightarrow \infty} \lVert S_N \rVert
    =
    \infty .
  \end{equation*}
  Due to the Banach--Steinhaus theorem \cite[p. 98]{rudin87_book} there exists a
  $\Ft{h}_{T_{*1}} \in C[-\pi,\pi]$ such that
  \begin{equation}\label{94}
    \limsup_{N \rightarrow \infty} \left( \max_{\omega \in [-\pi,\pi]} \left|
        \sum_{n=1}^{N} w_n(\Ft{h}_{T_{*1}},0) \Ft{g}_n(\omega)
      \right| \right)
    =
    \infty .
  \end{equation}
  Since
  \begin{align*}
    &\sum_{n=1}^{N} \left( \frac{1}{2\pi} \int_{-\pi}^{\pi} \Ft{f}(\omega) \Ft{g}_n(\omega)
      \e^{\iu \omega t} \di{\omega} \right) w_n(\Ft{h}_{T_{*1}},0) \\
    &=
    \frac{1}{2\pi} \int_{-\pi}^{\pi} \Ft{f}(\omega) \e^{\iu \omega t} \left( \sum_{n=1}^{N}
      w_n(\Ft{h}_{T_{*1}},0) \Ft{g}_n(\omega) \right) \di{\omega} ,
  \end{align*}
  and
  \begin{align*}
    &\sup_{\lVert f \rVert_{\PWs_{\pi}^{1}} \leq 1}
    \sum_{n=1}^{N} \left( \frac{1}{2\pi} \int_{-\pi}^{\pi} \Ft{f}(\omega) \Ft{g}_n(\omega)
      \e^{\iu \omega t} \di{\omega} \right) w_n(\Ft{h}_{T_{*1}},0) \\
    &=
    \max_{\omega \in [-\pi,\pi]} \left|
        \sum_{n=1}^{N} w_n(\Ft{h}_{T_{*1}},0) \Ft{g}_n(\omega)
      \right| ,
  \end{align*}
  it follows from \eqref{94} and the Banach--Steinhaus theorem \cite[p.
  98]{rudin87_book} that there exists an $f_{*1} \in \PWs_{\pi}^{1}$ such that
  \eqref{91} is true.

  Now we prove assertion \ref{92}.
  For $\Ft{h}_T \in C[-\pi,\pi]$, it follows for fixed $t \in \R$ that $\Ft{h}_T(\omega) \e^{\iu
    \omega t}$ is a continuous function on $[-\pi,\pi]$, and hence the integral 
  \eqref{88} exists.
  Let $t \in \R$ be arbitrary but fixed, and let  $\Ft{h}_{T_{*1}} \in C[-\pi,\pi]$ be the
  function from \eqref{94}. We define
  \begin{equation*}
    \Ft{h}_{T_{*2}}(\omega)
    =
    \e^{-\iu \omega t} \Ft{h}_{T_{*1}}(\omega), \quad \omega \in [-\pi,\pi] ,
  \end{equation*}
  and clearly we have $\Ft{h}_{T_{*2}} \in C[-\pi,\pi]$.
  It follows that
  \begin{equation*}
    \sum_{n=1}^{N} w_n(\Ft{h}_{T_{*2}},t) \Ft{g}_n(\omega)
    =
    \sum_{n=1}^{N} w_n(\Ft{h}_{T_{*1}},0) \Ft{g}_n(\omega)
  \end{equation*}
  for all $\omega \in [-\pi,\pi]$ and all $N \in \N$.
  Hence, we see from \eqref{94} that
  \begin{equation*}
    \limsup_{N \rightarrow \infty} \left( \max_{\omega \in [-\pi,\pi]} \left|
      \sum_{n=1}^{N} w_n(\Ft{h}_{T_{*2}},t) \Ft{g}_n(\omega)
    \right| \right)
    =
    \infty ,
  \end{equation*}
  and, by the same reasoning that was used in the proof of assertion \ref{90}, there exists an
  $f_{*2} \in \PWs_{\pi}^{1}$ such that
  \eqref{93} is true. \qed
\end{proof}

\begin{remark}
  Clearly, the development of an implementation theory, as outlined in the introduction,
  is a challenging task. Some results are already known. For example, in \cite{boche10f}
  it was shown that for bounded bandlimited signals a low-pass filter cannot be
  implemented as a linear system, but only as a non-linear system. Further, problems that
  arise due to causality constraints were discussed in \cite{pohl09_book}.

  At this point, it is worth noting that Arnol'd's \cite{arnold57} and Kolmogorov's
  \cite{kolmogorov57} solution of Hilbert's thirteenth problem \cite{hilbert02} give
  another implementation for the analog computation of functions. For a discussion of the
  solution in the context of communication networks, we would like to refer the reader to
  \cite{goldenbaum13}.

  Finally, it would also be interesting to connect the ideas of this work with Feynman's
  ``Physics of Computation'' \cite{feynman99_book} and Landauer's principle
  \cite{landauer61,landauer96}. Right now we are at the beginning of this development.
\end{remark}

{
\raggedleft
"Wir, so gut es gelang,
haben das Unsre [(vorerst)] getan."\\
Friedrich H\"olderlin "Der Gang aufs Land - An Landauer"

}

\begin{acknowledgement}
  The authors would like to thank Ingrid Daubechies for valuable discussions of
  Conjectures~\ref{16} and \ref{18} and for pointing out connections to frame
  theory at the Strobl'11 conference and the ``Applied Harmonic Analysis and Sparse
  Approximation'' workshop at the Mathematisches Forschungsinstitut Oberwolfach in 2012.
  Further, the authors are thankful to Przemys{\l}aw \mbox{Wojtaszczyk} and Yurii
  Lyubarskii for valuable discussions of Conjecture~\ref{16} at the Strobl'11
  conference, and Joachim Hagenauer and Sergio Verd\'u for drawing our attention to
  \cite{butzer11a} and for discussions of related topics. We would also like to thank
  Mario Goldenbaum for carefully reading the manuscript and providing helpful comments.
\end{acknowledgement}


\begin{thebibliography}{10}
\providecommand{\url}[1]{{#1}}
\providecommand{\urlprefix}{URL }
\expandafter\ifx\csname urlstyle\endcsname\relax
  \providecommand{\doi}[1]{DOI~\discretionary{}{}{}#1}\else
  \providecommand{\doi}{DOI~\discretionary{}{}{}\begingroup
  \urlstyle{rm}\Url}\fi

\bibitem{arnold57}
Arnol'd, V.I.: On the representability of a function of two variables in the
  form {$\chi[\phi(x)+\psi(y)]$}.
\newblock Uspekhi Mat. Nauk \textbf{12}(2(74)), 119--121 (1957)

\bibitem{boas54_book}
Boas, R.P.: Entire Functions.
\newblock Academic Press (1954)

\bibitem{boche08c}
Boche, H., M{\"o}nich, U.J.: Time domain representation of systems on
  bandlimited signals.
\newblock In: Proceedings of the 2008 IEEE Information Theory Workshop
  (ITW'08), pp. 51--55 (2008)

\bibitem{boche10g}
Boche, H., M{\"o}nich, U.J.: Sampling-type representations of signals and
  systems.
\newblock Sampling Theory in Signal and Image Processing \textbf{9}(1--3),
  119--153 (2010)

\bibitem{boche11a}
Boche, H., M{\"o}nich, U.J.: Sampling of deterministic signals and systems.
\newblock {IEEE} Transactions on Signal Processing \textbf{59}(5), 2101--2111
  (2011).

\bibitem{boche14a_submitted}
Boche, H., M{\"o}nich, U.J.: No-go theorem for sampling-based signal
  processing.
\newblock In: Proceedings of the {IEEE} International Conference on Acoustics,
  Speech, and Signal Processing ({ICASSP} '14) (2014).
\newblock Accepted

\bibitem{boche14b_submitted}
Boche, H., M{\"o}nich, U.J.: System approximation with general measurement
  functionals.
\newblock In: Proceedings of the {IEEE} International Conference on Acoustics,
  Speech, and Signal Processing ({ICASSP} '14) (2014).
\newblock Accepted

\bibitem{boche10f}
Boche, H., M{\"o}nich, U.J., Kortke, A., Keusgen, W.: No-go theorem for linear
  systems on bounded bandlimited signals.
\newblock {IEEE} Transactions on Signal Processing \textbf{58}(11), 5639--5654
  (2010).

\bibitem{butzer94}
Butzer, P.L.: The {H}ausdorff--{Y}oung theorems of {F}ourier analysis and their
  impact.
\newblock Journal of Fourier Analysis and Applications \textbf{1}(2), 113--130
  (1994).

\bibitem{butzer11a}
Butzer, P.L., Dodson, M.M., Ferreira, P.J.S.G., Higgins, J.R., Lange, O.,
  Seidler, P., Stens, R.L.: Multiplex signal transmission and the development
  of sampling techniques: the work of {H}erbert {R}aabe in contrast to that of
  {C}laude {S}hannon.
\newblock Applicable Analysis \textbf{90}(3--4), 643--688 (2011).

\bibitem{butzer11}
Butzer, P.L., Ferreira, P.J.S.G., Higgins, J.R., Saitoh, S., Schmeisser, G.,
  Stens, R.L.: Interpolation and sampling: {E.T.} {W}hittaker, {K.} {O}gura and
  their followers.
\newblock Journal of Fourier Analysis and Applications \textbf{17}(2), 320--354
  (2011).

\bibitem{butzer98}
Butzer, P.L., Lei, J.: Errors in truncated sampling series with measured
  sampled values for not-necessarily bandlimited functions.
\newblock Functiones et Approximatio Commentarii Mathematici \textbf{26},
  25--39 (1998)

\bibitem{butzer00}
Butzer, P.L., Lei, J.: Approximation of signals using measured sampled values
  and error analysis.
\newblock Communications in Applied Analysis. An International Journal for
  Theory and Applications \textbf{4}(2), 245--255 (2000)

\bibitem{butzer12}
Butzer, P.L., Schmeisser, G., Stens, R.L.: {S}hannon's sampling theorem for
  bandlimited signals and their {H}ilbert transform, {B}oas-type formulae for
  higher order derivative---the aliasing error involved by their extensions
  from bandlimited to non-bandlimited signals.
\newblock Entropy \textbf{14}(11), 2192--2226 (2012).

\bibitem{butzer80a}
Butzer, P.L., Splettst{\"o}{\ss}er, W.: On quantization, truncation and jitter
  errors in the sampling theorem and its generalizations.
\newblock Signal Processing \textbf{2}(2), 101--112 (1980)

\bibitem{butzer88}
Butzer, P.L., Splettst{\"o}{\ss}er, W., Stens, R.L.: The sampling theorem and
  linear prediction in signal analysis.
\newblock Jahresbericht der Deutschen Mathematiker-Vereinigung \textbf{90}(1),
  1--70 (1988)

\bibitem{butzer92}
Butzer, P.L., Stens, R.L.: Sampling theory for not necessarily band-limited
  functions: A historical overview.
\newblock SIAM Review \textbf{34}(1), 40--53 (1992).

\bibitem{ferreira95}
Ferreira, P.J.S.G.: Nonuniform sampling of nonbandlimited signals.
\newblock {IEEE} Signal Processing Letters \textbf{2}(5), 89--91 (1995).

\bibitem{ferreira99}
Ferreira, P.J.S.G.: Nonlinear systems and exponential eigenfunctions.
\newblock {IEEE} Signal Processing Letters \textbf{6}(11), 287--289 (1999).

\bibitem{ferreira01}
Ferreira, P.J.S.G.: Sorting continuous-time signals: analog median and
  median-type filters.
\newblock {IEEE} Transactions on Signal Processing \textbf{49}(11), 2734--2744
  (2001).

\bibitem{feynman99_book}
Feynman, R.P.: Feynman Lectures on Computation.
\newblock Penguin Books (1999)

\bibitem{fine49}
Fine, N.J.: On the {W}alsh functions.
\newblock Transactions of the American Mathematical Society \textbf{65},
  372--414 (1949).

\bibitem{franks69_book}
Franks, L.: Signal Theory.
\newblock Prentice Hall (1969)

\bibitem{goldenbaum13}
Goldenbaum, M., Boche, H., Stanczak, S.: Harnessing interference for analog
  function computation in wireless sensor networks.
\newblock {IEEE} Transactions on Signal Processing \textbf{61}(20), 4893--4906
  (2013).

\bibitem{heil11_book}
Heil, C.: A Basis Theory Primer: Expanded Edition, \emph{Applied and Numerical
  Harmonic Analysis}, vol.~1.
\newblock Birkh{\"a}user Boston (2011)

\bibitem{higgins85}
Higgins, J.R.: Five short stories about the cardinal series.
\newblock Bull. Amer. Math. Soc. \textbf{12}(1), 45--89 (1985)

\bibitem{higgins96_book}
Higgins, J.R.: Sampling Theory in {F}ourier and Signal Analysis -- Foundations.
\newblock Oxford University Press (1996)

\bibitem{hilbert02}
Hilbert, D.: Mathematical problems.
\newblock Bulletin of the American Mathematical Society \textbf{8}, 437--479
  (1902).

\bibitem{jerri77}
Jerri, A.J.: The {S}hannon sampling theorem--its various extensions and
  applications: A tutorial review.
\newblock Proceedings of the {IEEE} \textbf{65}(11), 1565--1596 (1977)

\bibitem{kolmogorov57}
Kolmogorov, A.N.: On the representation of continuous functions of many
  variables by superposition of continuous functions of one variable and
  addition.
\newblock Doklady Akademii Nauk SSSR \textbf{114}, 953--956 (1957)

\bibitem{landauer61}
Landauer, R.: Irreversibility and heat generation in the computing process.
\newblock IBM Journal of Research and Development \textbf{5}(3), 183--191
  (1961).

\bibitem{landauer96}
Landauer, R.: The physical nature of information.
\newblock Physics Letters A \textbf{217}(4--5), 188--193 (1996).

\bibitem{levin96_book}
Levin, B.Y.: Lectures on Entire Functions.
\newblock AMS (1996)

\bibitem{marvasti01_book}
Marvasti, F. (ed.): Nonuniform Sampling: Theory and Practice.
\newblock Kluwer Academic / Plenum Publishers (2001)

\bibitem{monich11_phdthesis}
M{\"o}nich, U.J.: Reconstruction and processing of bandlimited signals based on
  their discrete values.
\newblock Ph.D. thesis, Technische Universit{\"a}t M{\"u}nchen, Munich, Germany
  (2011)

\bibitem{olevskii66a}
Olevskii, A.M.: {F}ourier series of continuous functions with respect to
  bounded orthonormal systems.
\newblock Izv. Akad. Nauk SSSR Ser. Mat. \textbf{30}(2), 387--432 (1966)

\bibitem{olevskii66}
Olevskii, A.M.: An orthonormal system and its applications.
\newblock Mat. Sb. (N.S.) \textbf{71}(113), 297--336 (1966)

\bibitem{olevskii75_book}
Olevskii, A.M.: Fourier Series with Respect to General Orthogonal Systems,
  \emph{{E}rgebnisse der {M}athematik und ihrer {G}renzgebiete. 2. Folge},
  vol.~86.
\newblock Springer-Verlag (1975)

\bibitem{oppenheim09_book}
Oppenheim, A.V., Schafer, R.W.: Discrete-Time Signal Processing, 3 edn.
\newblock Prentice Hall (2009)

\bibitem{partington96}
Partington, J.R.: Recovery of functions by interpolation and sampling.
\newblock Journal of Mathematical Analysis and Applications \textbf{198}(2),
  301--309 (1996).

\bibitem{pavlov79}
Pavlov, B.S.: Basicity of an exponential system and {M}uckenhoupt's condition.
\newblock Dokl. Akad. Nauk SSSR \textbf{247}(1), 37--40 (1979).
\newblock English translation in Sov. Math. Dokl. 20 (1979), no. 4, 655--659.

\bibitem{pohl09_book}
Pohl, V., Boche, H.: Advanced Topics in System and Signal Theory: A
  Mathematical Approach, \emph{Foundations in Signal Processing, Communications
  and Networking}, vol.~4.
\newblock Springer (2009)

\bibitem{rudin87_book}
Rudin, W.: Real and Complex Analysis, 3 edn.
\newblock McGraw-Hill (1987)

\bibitem{schipp90_book}
Schipp, F., Wade, W.R., Simon, P.: {W}alsh Series: An Introduction to Dyadic
  Harmonic Analysis.
\newblock Adam Hilger (1990)

\bibitem{seip98}
Seip, K.: Developments from nonharmonic {F}ourier series.
\newblock In: Documenta Mathematica, Proc. ICM, vol.~II, pp. 713--722 (1998)

\bibitem{shannon49}
Shannon, C.E.: Communication in the presence of noise.
\newblock In: Proceedings of the IRE, vol.~37, pp. 10--21 (1949)

\bibitem{song06}
Song, Z., Yang, S., Zhou, X.: Approximation of signals from local averages.
\newblock Applied Mathematics Letter \textbf{19}(12), 1414--1420 (2006).

\bibitem{stens83}
Stens, R.L.: A unified approach to sampling theorems for derivatives and
  {H}ilbert transforms.
\newblock Signal Processing \textbf{5}, 139--151 (1983)

\bibitem{sun02}
Sun, W., Zhou, X.: Average sampling in spline subspaces.
\newblock Applied Mathematics Letter \textbf{15}(2), 233--237 (2002).

\bibitem{sun02a}
Sun, W., Zhou, X.: Reconstruction of band-limited signals from local averages.
\newblock {IEEE} Transactions on Information Theory \textbf{48}(11), 2955--2963
  (2002).

\bibitem{sun03}
Sun, W., Zhou, X.: Average sampling in shift invariant subspaces with symmetric
  averaging functions.
\newblock Journal of Mathematical Analysis and Applications \textbf{287}(1),
  279--295 (2003).

\bibitem{szarek80}
Szarek, S.J.: Nonexistence of {B}esselian basis in {$C(S)$}.
\newblock Journal of Functional Analysis \textbf{37}, 56--67 (1980)

\bibitem{young01_book}
Young, R.M.: An Introduction to Nonharmonic {F}ourier Series.
\newblock Academic Press (2001)

\bibitem{zygmund02_vol1_book}
Zygmund, A.: Trigonometric Series, vol.~I, 3 edn.
\newblock Cambridge University Press (2002)

\end{thebibliography}
\end{document}